\documentclass[smallcondensed]{svjour3}     % onecolumn (ditto)
\smartqed  % flush right qed marks, e.g. at end of proof
\usepackage{mathptmx}      % use Times fonts if available on your TeX system
\usepackage{graphicx}
\usepackage{cite}
\usepackage[cmex10]{amsmath}
\usepackage{algorithmic,algorithm}
\usepackage{array}
\usepackage{mdwmath}
\usepackage{mdwtab}
\usepackage{eqparbox}
\usepackage[tight,footnotesize]{subfigure}
\usepackage{epsfig}
\usepackage{url}
\usepackage{fancyhdr}
\usepackage{mathrsfs}
\usepackage{amssymb}
\usepackage[pdftex,bookmarksnumbered,bookmarksopen]{hyperref}
\usepackage{makecell}
\newtheorem{condition}{Condition}%[section]

\newcommand{\rb}[1]{\raisebox{1.3ex}[0pt]{#1}}

\makeatletter
\newcommand{\rmnum}[1]{\romannumeral #1}
\newcommand{\Rmnum}[1]{\expandafter\@slowromancap\romannumeral #1@}
\makeatother
\journalname{Designs, Codes and Cryptography}
\begin{document}

\title{\mbox{C-Codes}: Cyclic \mbox{Lowest-Density} MDS Array Codes Constructed Using Starters for \mbox{RAID 6}\thanks{This paper is a revised and expanded version of the material that appeared at the 2010 IEEE International Symposium on Information Theory, Austin, TX, June 2010 (see \cite{cyclic-code-starter}). The main part of this work was carried out while Mingqiang Li was a Ph.D. student in the Department of Computer Science and Technology, Tsinghua University. This work was supported by the National Science Foundation for Distinguished Young Scholars of China (Grant No. 60925006) and the National High Technology Research and Development Program of China (Grant No. 2009AA01A403).}
}
%\subtitle{Do you have a subtitle?\\ If so, write it here}

%\titlerunning{Short form of title}        % if too long for running head

\author{Mingqiang Li \and Jiwu Shu}

%\authorrunning{Short form of author list} % if too long for running head

\institute{Mingqiang Li \at
              IBM China Research Laboratory, Diamond Building 19-A, Zhongguancun Software Park, Dongbeiwang West Road
No. 8, Shangdi, Haidian District, Beijing 100193, P.R.China \\
              Tel.: +86-10-58748904\\
              Fax: +86-10-58748230\\
              \email{mingqiangli@cn.ibm.com}           %  \\
%             \emph{Present address:} of F. Author  %  if needed
           \and
           Jiwu Shu \at
              Department of Computer Science and Technology, Tsinghua University, Qinghuayuan No. 1, Haidian District,
Beijing 100084, P.R.China\\
              Tel.: +86-10-62795215\\
              Fax: +86-10-62771138\\
              \email{shujw@tsinghua.edu.cn}           %  \\
%             \emph{Present address:} of F. Author  %  if needed
}

\date{Received: date / Accepted: date}
% The correct dates will be entered by the editor

\maketitle

\begin{abstract}
The \mbox{distance-3} cyclic \mbox{lowest-density} MDS array code (called the \mbox{C-Code}) is a good candidate for RAID 6 because of its optimal storage efficiency, optimal update complexity, optimal length, and cyclic symmetry. In this paper, the underlying connections between \mbox{C-Codes} (or \mbox{quasi-C-Codes}) and starters in group theory are revealed. It is shown that each \mbox{C-Code} (or \mbox{quasi-C-Code}) of length $2n$ can be constructed using an even starter (or even \mbox{multi-starter}) in $(Z_{2n},+)$. It is also shown that each \mbox{C-Code} (or \mbox{quasi-C-Code}) has a twin \mbox{C-Code} (or \mbox{quasi-C-Code}). Then, four infinite families (three of which are new) of \mbox{C-Codes} of length $p-1$ are constructed, where $p$ is a prime. Besides the family of length $p-1$, \mbox{C-Codes} for some sporadic even lengths are also presented. Even so, there are still some even lengths (such as 8) for which \mbox{C-Codes} do not exist. To cover this limitation, two infinite families (one of which is new) of \mbox{quasi-C-Codes} of length $2(p-1)$ are constructed for these even lengths.
\keywords{RAID 6 \and Array codes \and Starters \and Perfect \mbox{one-factorization}}
% \PACS{PACS code1 \and PACS code2 \and more}
\subclass{05C70 \and 20K01 \and 94B05}
\end{abstract}

\section{Introduction}\label{section-introduction}
Array codes \cite{array-code} are a class of linear codes whose information and parity bits are placed in a \mbox{two-dimensional} (or multidimensional) array rather than a \mbox{one-dimensional} vector. A common property of array codes is that they are implemented based on only simple \mbox{eXclusive-OR} (XOR) operations. This is an attractive advantage in contrast to the family of \mbox{Reed-Solomon} codes \cite{RS,RS-Cauchy2,RS-Cauchy1} whose encoding and decoding processes use complex \mbox{finite-field} operations. Thus, array codes are ubiquitous in data storage applications.\\
\indent Among all kinds of array codes, cyclic \mbox{lowest-density} \mbox{Maximum-Distance} Separable (MDS) array codes \cite{cyclic-code} are regarded as the optimal ones for data storage applications because they have all the following properties:
\begin{enumerate}
    \item they are MDS codes, which can tolerate the loss of any $t$ disks in a storage system with $t$ additional disks of parity data and thus have optimal \emph{storage efficiency} (i.e. the ratio of user data to the total of user data plus parity data);
    \item their \emph{update complexity} (defined as the average number of parity bits affected by a change of a single information bit) achieves the minimum update complexity that MDS codes can have; and
    \item their regularity in the form of cyclic symmetry makes their implementation simpler and potentially less costly.
\end{enumerate}

\begin{table}[H]
  \caption{Comparison among some representative MDS array codes for RAID 6.}
  \label{table-code-comparison}
%  \centering
  \begin{minipage}[t]{0.68\textwidth}
  \begin{center}
  \begin{tabular}{|c|c||c|c|c|}
    \hline
    % after \\: \hline or \cline{col1-col2} \cline{col3-col4} ...
    \multicolumn{2}{|c||}{\makecell{Array Codes}} & \makecell{Optimal Update\\ Complexity} & \makecell{Optimal \\ Length\footnote{Considered only for lowest-density MDS codes (see \cite{B-Code}).}} & \makecell{Cyclic\\ Symmetry} \\
    \hline
    \hline
     & EVENODD\cite{EVENODD} & & & \\
    \cline{2-2}
    Horizontal & RDP\cite{RDP} & No & --- & No  \\
    \cline{2-2}
    \rb{Codes} & \makecell{Minimum Density\\ \mbox{RAID-6} Codes\cite{MinimumDensityRAID6Codes}} & & & \\
    \hline
     & X-Code\cite{X-Code} & Yes & No & No  \\
    \cline{2-5}
    \rb{Vertical} & B-Code\footnote{Including ZZS Code\cite{ZZS} and P-Code\cite{P-Code}.}\cite{B-Code} & Yes & Yes & No \\
    \cline{2-5}
    \rb{Codes} & C-Code\footnote{Including the codes constructed in \cite{cyclic-code}.}  & Yes & Yes & Yes \\
    \hline
  \end{tabular}
  \end{center}
  \end{minipage}
\end{table}

\indent Fault tolerance is an important concern in the design of \mbox{disk-based} storage systems \cite{latent_sector_error,EVENODD,RAID,RDP,P-Code,GRID,DACO,SDC,disk_failure1,MinimumDensityRAID6Codes,dRAID2-IBM,disk_failure2,RAID-high-level-survey,dRAID1-IBM,X-Code}. As today's storage systems grow in size and complexity, they are increasingly confronted with disk failures \cite{disk_failure1,disk_failure2} together with latent sector errors \cite{latent_sector_error}. Then, RAID 5 \cite{RAID}, which has been widely used in modern storage systems to recover one disk failure, cannot provide sufficient reliability guarantee. This results in the demand of RAID 6 \cite{RAID,RAID-high-level-survey}, which can tolerate two disk failures.\\
\indent RAID 6 is designed based on a \mbox{distance-3} MDS linear code. In applications with many small writes, such as the \mbox{On-Line} Transaction Processing (OLTP) application, the \mbox{distance-3} cyclic \mbox{lowest-density} MDS array code (called the \mbox{C-Code}) defined in Section~\ref{subsection-definition} can be regarded as a good candidate for RAID 6 (see Table~\ref{table-code-comparison}). It should be noted that the \mbox{C-Code} is also a good candidate for distributed RAID (dRAID) \cite{dRAID2-IBM,dRAID1-IBM} of fault tolerance 2 because of its optimal storage efficiency, optimal update complexity, and cyclic symmetry. Under this background, we will make a systematic study on the \mbox{C-Code} in this paper.\\
\indent Here, the \mbox{distance-3} cyclic \mbox{lowest-density} MDS array code is one particular kind of the \mbox{B-Code} described in \cite{B-Code}. Its additional feature is the regularity in the form of cyclic symmetry. It is thus called the \mbox{C-Code} in this paper. As we know, it is often an interesting and worthwhile problem to develop an array code that is also a cyclic code. Although cyclic symmetry can make the implementation of an array code simpler, it can make the construction of the code more complicated at the same time. As will be shown in this paper, the constructions of \mbox{C-Codes} (or \mbox{quasi-C-Codes})  require more efforts than those of general \mbox{B-Codes}. Besides, although Cassuto and Bruck constructed one infinite family of \mbox{C-Code} instances of length $p-1$ and one infinite family of \mbox{quasi-C-Code} instances of length $2(p-1)$ for the first time in \cite{cyclic-code}, they did not study the general constructions of \mbox{C-Codes} (or \mbox{quasi-C-Codes}). So, in this paper, we will carry out a systematic study on the constructions of \mbox{C-Codes} by revealing the underlying connections between \mbox{C-Codes} (or \mbox{quasi-C-Codes}) and \emph{starters} \cite{starter} in group theory.

\begin{definition}[\cite{P1F}]\label{definition-P1F}
A \emph{\mbox{one-factorization}} of a graph is a partitioning of the set of its edges into subsets such that each subset is a graph of degree one. Here, each subset is called a \emph{\mbox{one-factor}}. A \emph{perfect \mbox{one-factorization}} (or \emph{P1F}) is a particular \mbox{one-factorization} in which the union of any pair of \mbox{one-factors} is a Hamiltonian cycle.
\end{definition}

\indent (\emph{Remark:} A Hamiltonian cycle is a cycle in an undirected graph, which visits each vertex exactly once and also returns to the starting vertex.)\\
\indent The \mbox{C-Code} can be described using a graph approach proposed in \cite{B-Code} (see Section~\ref{subsection-description}). It will be shown in Section~\ref{subsection-equivalence} that the constructions of the \mbox{C-Code} of length $2n$ (denoted by $\mathbb{C}_{2n}$) are equivalent to bipyramidal P1Fs of a \mbox{$2n$-regular} graph on $2n+2$ vertices. In the literature of graph theory, we noticed that several known P1Fs \cite{P1F-list-(even)-starter-K32,P1F-K36-starter,P1F-K36-even-starter,P1F-K40-even-starter,P1F-K50-starter,P1F-K52-even-starter,2-even-starter-induced-P1Fs} were constructed using starters \cite{starter} in group theory. Inspired by this, we immediately raise a question: \emph{Which kind of starter can be used to construct the \mbox{C-Code}?} In Section~\ref{section-c-code}, we will show that each $\mathbb{C}_{2n}$ instance can be constructed using an \emph{even starter} \cite{even-starter-first} in $(Z_{2n},+)$. The necessary and sufficient condition is that the even starter can induce a bipyramidal P1F of a \mbox{$2n$-regular} graph on $2n+2$ vertices. Then, we will obtained \mbox{C-Codes} for some sporadic even lengths listed as follows: \[4, 6, 10, 12, 14, 16, 18, 20, 22, 24, 26, 28, 30, 32, 34, 36, 38, 50.\] Among them, 14, 20, 24, 26, 32, 34, 38, and 50 are not covered by the family of length $p-1$ presented in \cite{cyclic-code} (where $p$ is a prime). We will also show that a \mbox{C-Code} exists for most but not all of even lengths (one exception we found is 8) and that there are often more than one \mbox{C-Code} instances for a given even length (each \mbox{C-Code} instance always has a \emph{twin \mbox{C-Code} instance}). Here, the twin code can provide some choice flexibility for RAID 6 design. Then, we may wonder if there exists an infinite family of \mbox{C-Codes}. In Section~\ref{section-c-code-family}, we will construct four infinite families of $\mathbb{C}_{p-1}$ instances (which cover the family of $\mathbb{C}_{p-1}$ instances constructed in \cite{cyclic-code}) from two infinite families of even starters in $\left(Z^*_{p},\times\right)$. We will also conclude that  \mbox{non-cyclic} \mbox{B-Codes} of length $p-1$ constructed in \cite{P-Code}, \cite{B-Code}, and \cite{ZZS} can always be transformed to $\mathbb{C}_{p-1}$ instances.\\
\indent Besides, we noticed that there is no \mbox{C-Code} for some even lengths, such as 8. Then, one question is: \emph{Can we construct \mbox{quasi-C-Codes} (which partially hold cyclic symmetry \cite{cyclic-code}) for these even lengths?} In this paper, we say a \mbox{quasi-C-Code} of length $2n$ is a \emph{\mbox{$\kappa$-quasi-C-Code}} (denoted by $\mathbb{C}_{2n}^\kappa$, where $\kappa|2n$) if for $i=0,1,\cdots,\kappa-1$, each group of $\frac{2n}{\kappa}$ columns \[C_{i+\kappa\times 0},C_{i+\kappa\times 1},\cdots,C_{i+\kappa\times(\frac{2n}{\kappa}-1)}\] hold cyclic symmetry, where $C_j$ represents the $j$-th column of the code. In Section~\ref{section-quasi-c-code}, we will introduce a concept of \emph{even \mbox{multi-starters}} and then discuss how to construct \mbox{quasi-C-Codes} using even \mbox{multi-starters} (with an example of length 8). We will also show that each \mbox{quasi-C-Code} instance has a \emph{twin \mbox{quasi-C-Code} instance}. Similarly, in Section~\ref{section-c-code-family-quasi}, we will present an infinite family of even \mbox{2-starter} in $\left(Z_{2(p-1)},+\right)$ and then construct two infinite families of $\mathbb{C}_{2(p-1)}^2$ instances (which cover the family of $\mathbb{C}_{2(p-1)}^2$ instances constructed in \cite{cyclic-code}) using this family of even \mbox{2-starter}. We will also conclude that \mbox{non-cyclic} \mbox{B-Codes} of length $2(p-1)$ constructed in \cite{B-Code} can always be transformed to $\mathbb{C}_{2(p-1)}^2$ instances.

\subsection{The Main Contributions of This Paper}
\indent The main new findings of this paper include:
\begin{enumerate}
  \item Each \mbox{C-Code} of length $2n$ can be constructed using an even starter in $(Z_{2n},+)$ (see Section~\ref{section-c-code}), while each \mbox{quasi-C-Code} of length $2n$ can be constructed using an even \mbox{multi-starter} in $(Z_{2n},+)$ (see Section~\ref{section-quasi-c-code});
  \item Each \mbox{C-Code} (or \mbox{quasi-C-Code}) has a twin \mbox{C-Code} (or \mbox{quasi-C-Code}) (see Sections~\ref{section-c-code}~and~\ref{section-quasi-c-code});
  \item A \mbox{C-Code} exists for most but not all of even lengths (one exception is 8) (see Section~\ref{section-c-code}); and
  \item \mbox{Non-cyclic} \mbox{B-Codes} of length $p-1$ constructed in \cite{P-Code}, \cite{B-Code}, and \cite{ZZS} can always be transformed to \mbox{C-Codes} (see Section~\ref{section-c-code-family}), while \mbox{non-cyclic} \mbox{B-Codes} of length $2(p-1)$ constructed in \cite{B-Code} can always be transformed to \mbox{quasi-C-Codes} (see Section~\ref{section-c-code-family-quasi}).
\end{enumerate}

\indent Meanwhile, the main contributions of this paper to the constructions of \mbox{C-Codes} and \mbox{quasi-C-Codes} include:
\begin{enumerate}
  \item Four infinite families of \mbox{C-Codes} of length $p-1$ (which cover the family of \mbox{C-Codes} of length $p-1$ constructed in \cite{cyclic-code}) are constructed from two infinite families of even starters in $\left(Z^*_{p},\times\right)$ (where $p$ is a prime) in Section~\ref{section-c-code-family};
  \item Besides the family of length $p-1$, \mbox{C-Codes} for some sporadic even lengths listed as follows are obtained in Section~\ref{section-c-code}: \[14, 20, 24, 26, 32, 34, 38, 50;\]
  \item Two infinite families of \mbox{quasi-C-Codes} of length $2(p-1)$ (which cover the family of \mbox{quasi-C-Codes} of length $2(p-1)$ constructed in \cite{cyclic-code}) are constructed using an infinite family of even \mbox{2-starter} in $\left(Z_{2(p-1)},+\right)$ in Section~\ref{section-c-code-family-quasi}.
\end{enumerate}

\indent We begin this paper with an introduction of the \mbox{C-Code} in the next section.

\section{An Introduction of the \mbox{C-Code}}
\subsection{Definition and Structure}\label{subsection-definition}
\indent The \mbox{C-Code} is one particular kind of the \mbox{B-Code} described in \cite{B-Code}. Its additional feature is the regularity in the form of cyclic symmetry, and its algebraic definition is given as follows:
\begin{definition}\label{definition-cyclic-code}
Let \[\mathbf{H}_{2n}=(H_0\ H_1\ \cdots\ H_{2n-1})\] be a binary matrix, where \[H_k=(h_{i,j})_{2n\times n}\] is a binary submatrix of size $2n\times n$, for $0\leq i\leq 2n-1$, $0\leq j\leq n-1$, and $0\leq k\leq 2n-1$. Suppose $\mathbf{H}_{2n}$ meets the following four conditions:
\begin{enumerate}
  \item for $k=0,1,\cdots,2n-1$, the last column of $H_k$ is the same as the \mbox{$k$-th} column of a binary $2n\times 2n$ identity matrix;
  \item for $k=0,1,\cdots,2n-1$,
  \begin{equation}
  H_k=E_{2n}^{k}\times H_0,
  \end{equation}
  where $E_{2n}$ is a binary \emph{elemental cyclic matrix} defined as
  \[
  E_{2n}=\left(
           \begin{array}{cc}
             \overrightarrow{0} & 1 \\
             I_{2n-1} & {\overrightarrow{0}}^T \\
           \end{array}
         \right),
  \]
  where $I_{2n-1}$ is a binary $(2n-1)\times(2n-1)$ identity matrix, $\overrightarrow{0}$ is a binary $1\times(2n-1)$ vector of 0's, and ${\overrightarrow{0}}^T$ is a binary $(2n-1)\times 1$ vector of 0's;
  \item the \emph{weight} (i.e. the number of 1's) of each row of $\mathbf{H}_{2n}$ is $2n-1$; and
  \item for any $m$ and $k$ (where $0\leq m<k\leq 2n-1$), the square matrix $(H_m\ H_k)$ is nonsingular.
\end{enumerate}
If a code's \mbox{parity-check} matrix is $\mathbf{H}_{2n}$, the code is then called the \mbox{C-Code} of length $2n$, denoted by $\mathbb{C}_{2n}$.
\end{definition}

\indent In the above definition, it should be noted that the length of the \mbox{C-Code} is always an even number. This is guaranteed by the MDS property of the \mbox{B-Code} \cite{B-Code}.\\
\indent Take $\mathbb{C}_4$ for example. The \mbox{parity-check} matrix for a $\mathbb{C}_4$ instance is as follows:
\[
\mathbf{H}_4=\left(\begin{array}{cc|cc|cc|cc}
               0 & 1 & 0 & 0 & 1 & 0 & 1 & 0 \\
               1 & 0 & 0 & 1 & 0 & 0 & 1 & 0 \\
               1 & 0 & 1 & 0 & 0 & 1 & 0 & 0 \\
               0 & 0 & 1 & 0 & 1 & 0 & 0 & 1 \\
             \end{array}\right).
\]

\indent It can be easily checked that \[\left\{I_{2n},E_{2n},E_{2n}^{2},\cdots,E_{2n}^{2n-1}\right\}\] forms a cyclic group with binary matrix multiplication. We also have
\begin{equation}
E_{2n}^{2n}=I_{2n}.
\end{equation}

Thus, the code defined in Definition~\ref{definition-cyclic-code} has the property of cyclic symmetry.\\
\indent In addition, it can also be deduced from the results in \cite{B-Code} that the \mbox{C-Code} defined in Definition~\ref{definition-cyclic-code} has all the following optimal properties:
\begin{enumerate}
  \item it is \mbox{Maximum-Distance} Separable (MDS);
  \item its update complexity is 2, which is the minimum update complexity that MDS codes of distance 3 can have; and
  \item it achieves the maximum length that MDS codes with optimal update complexity can have.
\end{enumerate}

\indent The structure of the \mbox{C-Code} of length $2n$ evolves from that of the \mbox{B-Code} described in \cite{B-Code}. This kind of array code has dimensions $n\times 2n$, i.e. $n$ rows and $2n$ columns, where $n$ is an integer not smaller than 2. It was proved in \cite{ZZS} that this size has optimal length. The first $n-1$ rows are information rows, and the last row is a parity row. In other words, the bits in the first $n-1$ rows are information bits, while those in the last row are parity bits. Because of the optimal update complexity, each information bit contributes to the calculation of (or is protected by) exactly 2 parity bits contained in other columns. Moreover, any two information bits do not contribute to the calculation of the same pair of parity bits. Here, we can see that since each column (corresponding to a disk) contains both data and parity, the \mbox{C-Code} belongs to the category of vertical array codes.\\
\indent Take the foregoing $\mathbb{C}_4$ instance for example. It has 2 rows and 4 columns. Its array representation is given as follows:
\begin{center}
  \begin{tabular}{|c|c|c|c|}
  \hline
  % after \\: \hline or \cline{col1-col2} \cline{col3-col4} ...
  $d_{1,2}$ & $d_{2,3}$ & $d_{3,0}$ & $d_{0,1}$ \\
  \hline
  \hline
  $p_0$ & $p_1$ & $p_2$ & $p_3$ \\
  \hline
  \end{tabular}
  ,
\end{center}
where $d_{i,j}$ ($0\leq i\neq j\leq 3$) represents a information bit that contributes to the calculation of (or is protected by) 2 parity bits $p_i$ and $p_j$. Then, take the parity bit $p_0$ for example. It can be calculated by $p_0=d_{3,0}+d_{0,1}$.

\subsection{Graph Description}\label{subsection-description}
In the \mbox{C-Code}, each information bit contributes to the calculation of (or is protected by) exactly 2 parity bits contained in other columns. Moreover, any two information bits do not contribute to the same pair of parity bits. Thus, a graph approach \cite{B-Code} can be used to describe the \mbox{C-Code}.\\
\indent In the graph description of the \mbox{C-Code}, each parity bit is represented by a vertex, and each information bit that contributes to the calculation of 2 parity bits is represented by the edge that connects the two corresponding vertices. Then, a \mbox{C-Code} of length $2n$ can be described by a \mbox{$(2n-2)$-regular} graph $G$ on $2n$ vertices. We label the $2n$ vertices with integers from 0 to $2n-1$ such that the \mbox{$i$-th} vertex ($i=0,1,\cdots,2n-1$) represents the parity bit contained in the \mbox{$i$-th} column of the \mbox{C-Code}. Then, for $i=0,1,\cdots,2n-1$, the \mbox{$i$-th} column of the \mbox{C-Code} can be represented by a set of $n-1$ edges, i.e.
\[
    C_i=\left\{ \{x_{i,1},y_{i,1}\},\{x_{i,2},y_{i,2}\},\cdots,\{x_{i,n-1},y_{i,n-1}\} \right\},
\]
where $\{x_{i,j},y_{i,j}\}$ ($j=1,2,\cdots,n-1$) is an edge corresponding to an information bit contained in the \mbox{$i$-th} column. According to the cyclic symmetry of the \mbox{C-Code}, for $i=0,1,\cdots,2n-1$, we have
\begin{equation}\label{equation-cyclic}
C_{i}=\left\{ \{x+i \bmod{2n},y+i \bmod{2n}\}: \{x,y\}\in C_0\right\}.
\end{equation}
Thus, in this paper, we sometimes use $C_0$ to simply represent a \mbox{C-Code}.\\
\indent Take the foregoing $\mathbb{C}_4$ instance for example. Figure~\ref{fig-C4-graph} shows its graph representation. The corresponding graph $G$ is a \mbox{2-regular} graph on a set of 4 vertices $\{0,1,2,3\}$. The four columns of the code can be represented by
\[
\left\{
  \begin{array}{l}
    C_0=\left\{ \{1,2\} \right\}; \\
    C_1=\left\{ \{2,3\} \right\}; \\
    C_2=\left\{ \{3,0\} \right\}; \\
    C_3=\left\{ \{0,1\} \right\}.
  \end{array}
\right.
\]
It is clear that $\left\{C_0,C_1,C_2,C_3\right\}$ meets Equation~(\ref{equation-cyclic}). Then, this $\mathbb{C}_4$ instance can be represented simply by $C_0=\left\{ \{1,2\} \right\}$.

\begin{figure}[H]
%\centering
\subfigure[]{
\includegraphics[width=1.4in]{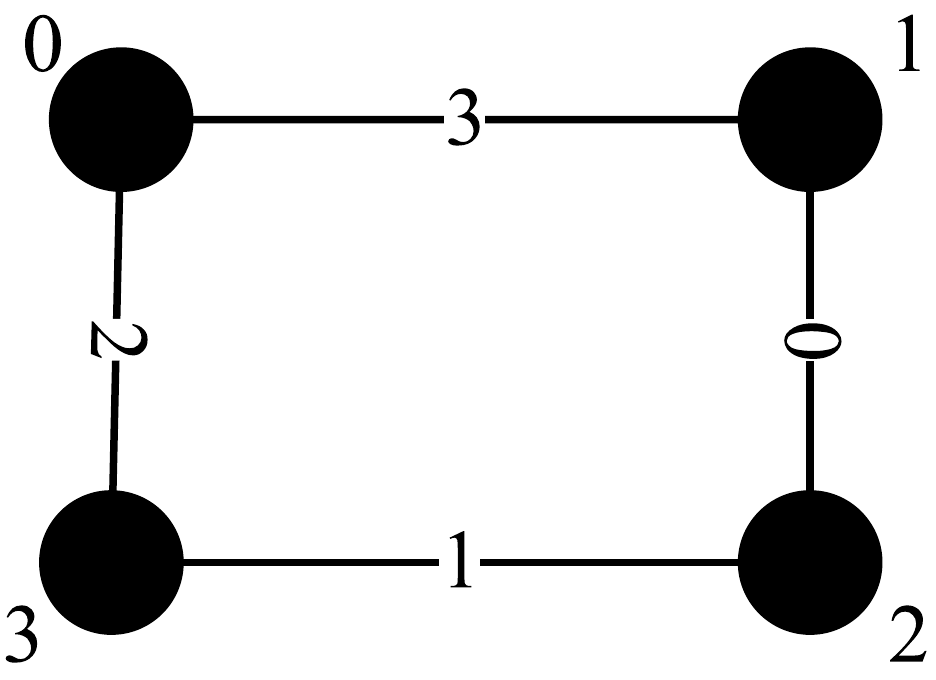}
\label{fig-C4-graph}
}
\hspace{0.2in}
\subfigure[]{
\includegraphics[width=1.4in]{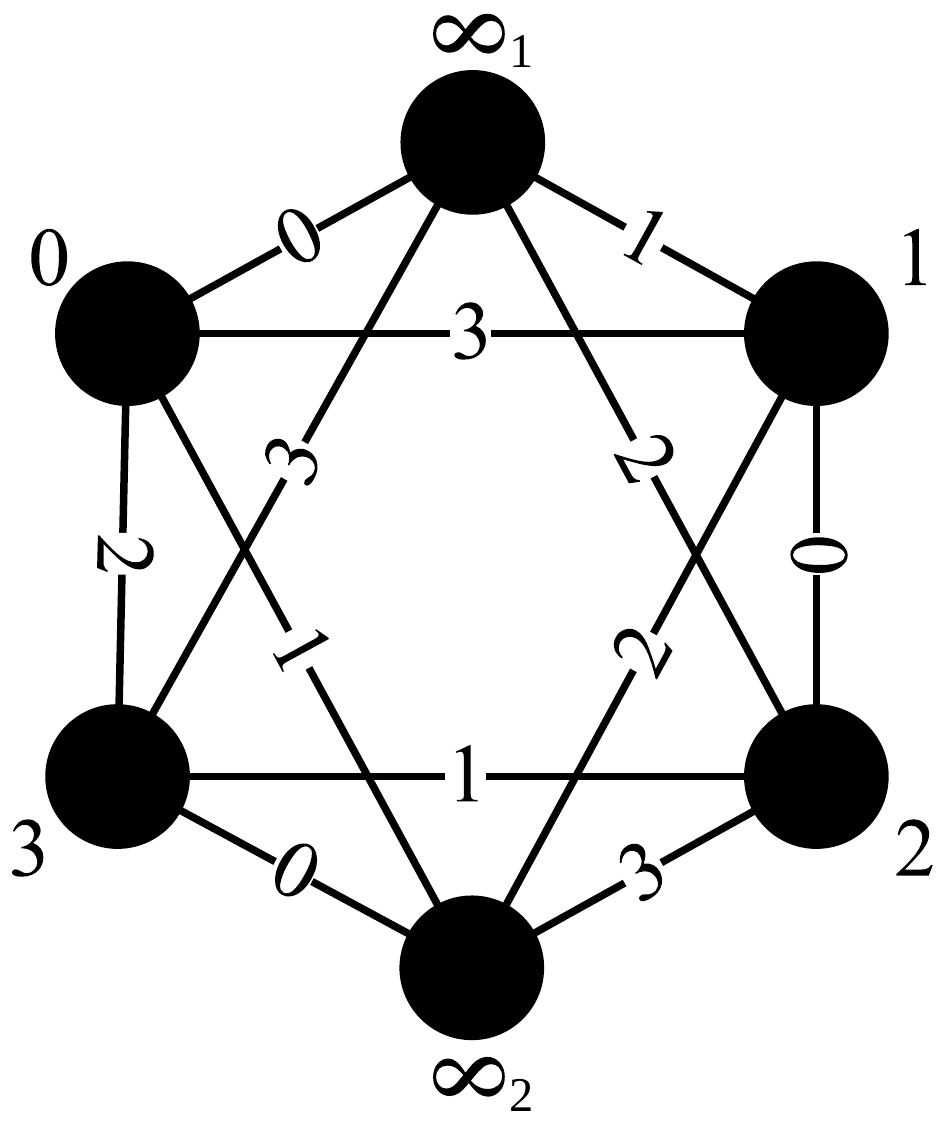}
\label{fig-expanded-C4-graph}
}
\caption{Constructing (a) a $\mathbb{C}_4$ instance from (b) a bipyramidal P1F of a \mbox{4-regular} graph on 6 vertices.}\label{fig-construction}
\end{figure}

\indent Recall the \mbox{C-Code} of length $2n$ can recover the erasure of any two columns. This is guaranteed by the fourth condition of Definition~\ref{definition-cyclic-code}. In the graph description, this condition is equivalent to the following one:
\begin{condition}\label{condition}
For any $m$ and $k$ (where $0\leq m<k\leq 2n-1$), the subgraph \[G^*=(\{0,1,\cdots,2n-1\},C_m\cup C_k)\] does not contain a cycle or a path whose terminal vertices are the two vertices $m$ and $k$.
\end{condition}

\indent The above condition is explained by contradiction as follows: \\
\indent We first consider the first opposite case where $G^*$ contains a cycle of length $r$. In such a cycle, suppose the $r$ edges are $e_1,e_2,\cdots,e_r$. As we know, in the corresponding square matrix mentioned in the fourth condition of Definition~\ref{definition-cyclic-code}, the column vector corresponding to each edge is a vector of weight 2, whose two 1's are in the two rows corresponding to the two vertices of the edge. Then, in the square matrix, the binary sum of the $r$ column vectors corresponding to $e_1,e_2,\cdots,e_r$ is a zero vertical vector, which conflicts with the nonsingular property of the square matrix. Thus, $G^*$ should not contain a cycle.\\
\indent We then consider the second opposite case where $G^*$ contains a path of length $r'$ whose terminal vertices are the two vertices $m$ and $k$. In such a path, suppose the $r'$ edges are ${e'}_1,{e'}_2,\cdots,{e'}_{r'}$. As we know, in the corresponding square matrix mentioned in the fourth condition of Definition~\ref{definition-cyclic-code}, the column vector corresponding to the terminal vertex $m$ (or $k$) is a vector of weight 1, whose only 1 is in the row corresponding to the vertex $m$ (or $k$). Then, in the square matrix, the binary sum of the $r'+2$ column vectors corresponding to the two terminal vertices $m$ and $k$ and the $r'$ edges ${e'}_1,{e'}_2,\cdots,{e'}_{r'}$ is a zero vertical vector, which conflicts with the nonsingular property of the square matrix. Thus, $G^*$ should not contain a path whose terminal vertices are the two vertices $m$ and $k$.

\subsection{The Equivalence Between \mbox{C-Code} Constructions and Bipyramidal P1Fs}\label{subsection-equivalence}
\begin{definition}\label{definition-cyclicbipyramidal1F}
For a \mbox{$2n$-regular} graph on $2n+2$ vertices $0,1,\cdots,2n-1,\infty_1,\infty_2$ (where the two vertices $\infty_1$ and $\infty_2$ are not adjacent to each other), a \emph{bipyramidal \mbox{one-factorization}} is a \mbox{one-factorization} consisting of $2n$
factors $F_0,F_1,\cdots,F_{2n-1}$, which are defined as
\begin{equation}
F_i=\left\{ \{\sigma^i(x),\sigma^i(y)\}:\{x,y\}\in F_0\right\}
\end{equation}
for $i=0,1,\cdots,2n-1$, where
\[\sigma=(0\ 1\ \cdots\ 2n-1)(\infty_1)(\infty_2)\] is a permutation represented by a product of disjoint cycles.
\end{definition}

\indent In the above definition, if the \mbox{one-factorization} is perfect, it is then called a \emph{bipyramidal P1F} of a \mbox{$2n$-regular} graph on $2n+2$ vertices.\\
\indent Then, we present the following theorem:
\begin{theorem}\label{theorem-equivalence}
The constructions of the \mbox{C-Code} of length $2n$ are equivalent to bipyramidal P1Fs of a \mbox{$2n$-regular} graph on $2n+2$ vertices. Suppose $\mathbb{F}$ is a bipyramidal P1F of a \mbox{$2n$-regular} graph on $2n+2$ vertices $0,1,\cdots,2n-1,\infty_1,\infty_2$ (where the two vertices $\infty_1$ and $\infty_2$ are not adjacent to each other), in which $F_0$ is the \mbox{one-factor} that contains the edge $\{0,\infty_1\}$. Then, the first column of the corresponding \mbox{C-Code} $\mathbb{C}_{2n}$ is
\begin{equation}
C_0=F_0\setminus\left\{ \{0,\infty_1\}, \{r,\infty_2\} \right\},
\end{equation}
where $r$ is the vertex that is adjacent to the vertex $\infty_2$ in $F_0$.
\end{theorem}
\begin{proof}
See Appendix~\ref{appendix1}.
\end{proof}

\section{Constructing a \mbox{C-Code} Using Even Starters}\label{section-c-code}
We first give the definition of an even starter \cite{even-starter-first} in $(Z_{2n},+)$.
\begin{definition}\label{definition-evenstarter}
An \emph{even starter} $S_E$ in $(Z_{2n},+)$ is a set of $n-1$ pairs of \mbox{non-zero} elements, i.e. \[S_E=\left\{
\{x_1,y_1\},\{x_2,y_2\},\cdots,\{x_{n-1},y_{n-1}\} \right\},\] such that for every element $i\in Z_{2n}$ such that $i\neq 0$ and $i\neq n$, there exists a pair $\{x,y\}\in S_E$ such that $i=(x-y) \bmod{2n}$, or $i=(y-x) \bmod{2n}$. Its \emph{twin even starter} $S_E^\tau$ is defined as
\begin{equation}\label{equation-twin-even-starter}
    S_E^\tau=\left\{\{x-r,y-r\}:\{x,y\}\in S_E\right\},
\end{equation}
where $r$ is the one and only \mbox{non-zero} element that does not occur in $S_E$.
\end{definition}

\indent Take $S_E=\left\{\{1,2\},\{3,5\} \right\}$ in $(Z_{6},+)$ for example. For every \mbox{non-zero} element except 3, we have
\[
\left\{
  \begin{array}{ll}
    1=2-1 \bmod{6}; \\
    2=5-3 \bmod{6}; \\
    4=3-5 \bmod{6}; \\
    5=1-2 \bmod{6}.
  \end{array}
\right.
\]
Thus, $S_E$ is an even starter in $(Z_{6},+)$. Its twin even starter is $S_E^\tau=\left\{\{3,4\},\{5,1\} \right\}$.\\
\indent Then, we present the following theorem:
\begin{theorem}\label{theorem-bipyramidal-1F-even-starter}
For a \mbox{$2n$-regular} graph on $2n+2$ vertices $0,1,\cdots,2n-1,\infty_1,\infty_2$ (where the two vertices $\infty_1$ and $\infty_2$ are not adjacent to each other), suppose a \mbox{one-factorization} $\mathbb{F}$ is a bipyramidal \mbox{one-factorization} in which $F_0$ is the \mbox{one-factor} that contains the edge $\{0,\infty_1\}$. Let
\begin{equation}
S=F_0\setminus\left\{ \{0,\infty_1\}, \{r,\infty_2\} \right\},
\end{equation}
where $r$ is the vertex that is adjacent to the vertex $\infty_2$ in $F_0$. Then, $S$ is an even starter in $(Z_{2n},+)$.
\end{theorem}
\begin{proof}
See Appendix~\ref{appendix2}.
\end{proof}

\indent According to Theorem~\ref{theorem-equivalence} in Section~\ref{subsection-equivalence}, we can further make the following conclusion:
\begin{theorem}\label{theorem-c-code-even}
In a \mbox{C-Code} $\mathbb{C}_{2n}$, the first column $C_0$ is always an even starter in $(Z_{2n},+)$. An even starter $S_E$ in $(Z_{2n},+)$ can be used to construct a \mbox{C-Code} of length $2n$ if and only if $S_E$ can induce a bipyramidal P1F of a \mbox{$2n$-regular} graph on $2n+2$ vertices.
\end{theorem}

\indent At the same time, we can deduce the following conclusion:

\begin{theorem}\label{theorem-twin}
If a $\mathbb{C}_{2n}$ instance can be constructed using an even starter $S_E$ in $(Z_{2n},+)$, another $\mathbb{C}_{2n}$  instance can also be constructed using the twin even starter $S_E^\tau$. They are called \emph{twin $\mathbb{C}_{2n}$ instances}.
\end{theorem}

\indent The above conclusion can be easily understood because twin even starters $S_E$ and $S_E^\tau$ induce the same bipyramidal \mbox{one-factorization}.\\
\indent Now, we discuss how to construct a \mbox{C-Code} $\mathbb{C}_{2n}$ using even starters in $(Z_{2n},+)$. The steps to construct $\mathbb{C}_{2n}$ are to first find an even starter $S_E$ in $(Z_{2n},+)$ and then check whether the bipyramidal \mbox{one-factorization} $\mathbb{F}$ induced by $S_E$ is a P1F of a \mbox{$2n$-regular} graph on a set of $2n+2$ vertices $0,1,\cdots,2n-1,\infty_1,\infty_2$ (where the two vertices $\infty_1$ and $\infty_2$ are not adjacent to each other). Here, let \[\mathbb{F}=\{F_0,F_1,\cdots,F_{2n-1}\}.\] According to the cyclic symmetry of $\mathbb{F}$, it is clear that if $F_0\cup F_i$ is a Hamiltonian cycle for all $i$ from 1 to $n$, $\mathbb{F}$ is then a P1F. Thus, only $n$ rather than ${2n \choose 2}$ subgraphs need to be checked in determining whether $\mathbb{F}$ is a P1F. According to Theorem~\ref{theorem-c-code-even}, if $\mathbb{F}$ is a P1F, a $\mathbb{C}_{2n}$ instance
\begin{equation}
C_0=S_E
\end{equation}
can be constructed; otherwise we try other even starters in $(Z_{2n},+)$ until a \mbox{C-Code} is constructed, or all even starters in $(Z_{2n},+)$ have been checked.\\
\indent For example, the foregoing even starter $S_E=\left\{\{1,2\},\{3,5\} \right\}$ in $(Z_{6},+)$ induces a P1F of a \mbox{6-regular} graph on 8 vertices. Thus, a $\mathbb{C}_6$ instance illustrated as follows can be constructed using $S_E$:
\begin{center}
  \begin{tabular}{|c|c|c|c|c|c|}
  \hline
  % after \\: \hline or \cline{col1-col2} \cline{col3-col4} ...
  $d_{1,2}$ & $d_{2,3}$ & $d_{3,4}$ & $d_{4,5}$ & $d_{5,0}$ & $d_{0,1}$ \\
  \hline
  $d_{3,5}$ & $d_{4,0}$ & $d_{5,1}$ & $d_{0,2}$ & $d_{1,3}$ & $d_{2,4}$ \\
  \hline
  \hline
  $p_0$ & $p_1$ & $p_2$ & $p_3$ & $p_4$ & $p_5$ \\
  \hline
  \end{tabular}
  .
\end{center}
At the same time, we can construct the twin $\mathbb{C}_6$ instance illustrated as follows using the twin even starter $S_E^\tau=\left\{\{3,4\},\{5,1\} \right\}$:
\begin{center}
  \begin{tabular}{|c|c|c|c|c|c|}
  \hline
  % after \\: \hline or \cline{col1-col2} \cline{col3-col4} ...
  $d_{3,4}$ & $d_{4,5}$ & $d_{5,0}$ & $d_{0,1}$ & $d_{1,2}$ & $d_{2,3}$ \\
  \hline
  $d_{5,1}$ & $d_{0,2}$ & $d_{1,3}$ & $d_{2,4}$ & $d_{3,5}$ & $d_{4,0}$ \\
  \hline
  \hline
  $p_0$ & $p_1$ & $p_2$ & $p_3$ & $p_4$ & $p_5$ \\
  \hline
  \end{tabular}
  .
\end{center}

\begin{table}[H]
%\centering
  \caption{Some examples of \mbox{C-Codes} for even lengths from 4 to 36.}\label{table-example}
    \begin{tabular}{|c|p{3.8in}|}
      \hline
      Length & $C_0$ (i.e. $S_E$) \\
      \hline
      \hline
      4 & $\left\{ \{1,2\} \right\}$ \\
      \hline
      6 & $\left\{ \{1,2\},\{3,5\} \right\}$ \\
      \hline
      8 & (Not Exist!) \\
      \hline
      10 & $\left\{ \{1,2\},\{3,5\},\{4,8\},\{6,9\} \right\}$ \\
      \hline
      12 & $\left\{ \{1,10\},\{2,6\},\{3,5\},\{4,9\},\{7,8\} \right\}$ \\
      \hline
      14 & $\left\{ \{1,2\},\{3,11\},\{4,6\},\{5,9\},\{7,10\},\{8,13\} \right\}$ \\
      \hline
      16 & $\left\{ \{1,2\},\{3,13\},\{4,15\},\{5,14\},\{6,8\},\{7,11\},\{9,12\}\right\}$ \\
      \hline
      18 & $\left\{ \{1,2\},\{3,7\},\{4,11\},\{5,15\},\{6,9\},\{8,13\},\{10,16\},\{12,14\} \right\}$ \\
      \hline
      20 & $\left\{ \{1,2\},\{3,5\},\{4,17\},\{6,14\},\{7,18\},\{8,13\},\{9,12\},\{10,16\},\{11,15\}\right\}$ \\
      \hline
      22 & $\left\{ \{1,2\},\{3,6\},\{4,12\},\{5,9\},\{7,13\},\{8,21\},\{10,20\},\{11,18\},\{14,19\}, \{15,17\}\right\}$ \\
      \hline
      24 & $\left\{ \{1,2\},\{3,5\},\{4,21\},\{6,11\},\{7,20\},\{8,12\},\{9,19\},\{10,16\},\{13,22\}, \{14,17\},\right.$
      $\left. \{15,23\}\right\}$ \\
      \hline
      26 & $\left\{ \{1,2\},\{3,6\},\{4,25\},\{5,19\},\{7,14\},\{8,24\},\{9,11\},\{10,18\},\{12,23\}, \{13,22\},\right.$
      $\left. \{15,21\},\{16,20\} \right\}$ \\
      \hline
      28 & $\left\{ \{1,2\},\{3,6\},\{4,25\},\{5,21\},\{7,11\},\{8,16\},\{9,18\},\{10,27\},\{12,22\}, \{13,26\},\right.$
      $\left. \{14,20\},\{15,17\},\{19,24\}\right\}$ \\
      \hline
      30 & $\left\{ \{1,2\},\{3,5\},\{4,9\},\{6,25\},\{7,13\},\{8,21\},\{10,24\},\{11,29\},\{12,16\},\{14,23\},\right.$
      $\left.\{15,22\},\{17,20\},\{18,28\}, \{19,27\} \right\}$ \\
      \hline
      32 & $\left\{ \{1,2\},\{3,5\},\{4,8\},\{6,27\},\{7,24\},\{9,21\},\{10,19\},\{11,29\},\{12,31\},\{13,18\},\right.$ $\left.\{14,17\},\{15,25\},\{16,22\},\{20,28\},\{23,30\} \right\}$ \\
      \hline
      34 & $\left\{ \{1,2\},\{3,5\},\{4,10\},\{6,25\},\{7,14\},\{8,32\},\{9,18\},\{11,22\},\{12,20\},\{13,26\},\right.$ $\left.\{15,33\},\{16,30\},\{17,21\},\{19,31\},\{23,28\},\{24,27\} \right\}$ \\
      \hline
      36 & $\left\{ \{1,2\},\{3,5\},\{4,8\},\{6,11\},\{7,20\},\{9,18\},\{10,34\},\{12,26\},\{13,28\},\{14,33\},\right.$ $\left.\{15,35\},\{16,22\},\{17,25\},\{19,29\},\{21,32\},\{23,30\},\{24,27\}\right\}$\\
      \hline
    \end{tabular}
\end{table}

\indent Finally, in the literature of graph theory, some bipyramidal P1Fs of a complete graph on $2n+2$ vertices (denoted by $K_{2n+2}$), which are induced by even starters in $(Z_{2n},+)$, have been found for the following values of $2n$: 4, 6, 10, 12, 14, 16, 18, 20, 22, 24, 26, 28, 30, 34, 38, and 50 \cite{P1F-list-(even)-starter-K32,P1F-K36-even-starter,P1F-K40-even-starter,P1F-K52-even-starter}. Here, a bipyramidal P1F $\mathbb{F}$ of $K_{2n+2}$ is induced by an even starter $S_E$ in $(Z_{2n},+)$ as follows: \[\mathbb{F}=\left\{F_0,F_1,\cdots,F_{2n-1},F_{2n}\right\},\] where $\left\{F_0,F_1,\cdots,F_{2n-1}\right\}$ is a bipyramidal P1F of a \mbox{$2n$-regular} graph on $2n+2$ vertices induced by $S_E$, and \[F_{2n}=\left\{\{0,n\},\{1,n+1\},\cdots,\{n-1,2n-1\},\{\infty_1,\infty_2\}\right\}.\] Thus, \mbox{C-Codes} for these values of $2n$ can be constructed using the corresponding even starters presented in \cite{P1F-K36-even-starter}, \cite{P1F-list-(even)-starter-K32}, \cite{P1F-K40-even-starter}, and \cite{P1F-K52-even-starter}. For example, a $\mathbb{C}_{50}$ instance can be constructed using the even starter presented in \cite{P1F-K52-even-starter}:
\begin{equation*}
    \begin{split}
      C_0= & \left\{\{2,29\}, \{3,35\}, \{4,16\}, \{5,33\}, \{6,43\}, \{7,15\},\right.\\
           & \left. \{8,19\}, \{9,30\}, \{10,41\}, \{11,46\}, \{12,17\}, \right.\\
           & \left. \{13,20\}, \{14,28\}, \{18,38\}, \{21,27\}, \{22,23\},  \right.\\
           & \left. \{24,48\},\{25,34\}, \{26, 36\}, \{31, 47\}, \{32, 49\}, \right.\\
           & \left. \{37, 39\}, \{40, 44\}, \{42, 45\} \right\}.
    \end{split}
\end{equation*}
Then, \mbox{C-Codes} for lengths 14, 20, 24, 26, 34, 38, and 50, which are not covered by the family of length $p-1$ presented in \cite{cyclic-code} (where $p$ is a prime), can be constructed here.\\
\indent It should be noted that an exhaustive search showed that a \mbox{C-Code} exists for most but not all of even lengths. For example, a \mbox{C-Code} exists for every even length from 4 to 36 except 8 (see Table~\ref{table-example}). Here, the length 32 is also not covered by the family of length $p-1$ presented in \cite{cyclic-code}. In the exhaustive search, for $2n=4,6,\cdots,36$, we first found an even starter $S_E$ in $(Z_{2n},+)$ according to Definition~\ref{definition-evenstarter} and then checked whether the bipyramidal \mbox{one-factorization} $\mathbb{F}$ induced by $S_E$ is a P1F of a \mbox{$2n$-regular} graph on $2n+2$ vertices using the foregoing efficient method presented in this section. If $\mathbb{F}$ is a P1F, a $\mathbb{C}_{2n}$ instance was constructed according to Theorem~\ref{theorem-c-code-even}; otherwise we tried other even starters in $(Z_{2n},+)$ until a \mbox{C-Code} was constructed, or all even starters in $(Z_{2n},+)$ had been checked (e.g. in the case of $2n=8$). \\
\indent It should also be noted that there are often more than one \mbox{C-Code} instances for a given even length. For example, besides the $\mathbb{C}_{34}$ instance listed in Table~\ref{table-example}, another $\mathbb{C}_{34}$ instance can be constructed using the even starter presented in \cite{P1F-K36-even-starter}:
\begin{equation*}
    \begin{split}
      C_0= & \left\{\{1,2\}, \{3, 5\}, \{4, 24\}, \{6, 9\}, \{7, 22\}, \{8, 18\}, \right.\\
           & \left. \{10, 17\}, \{12, 25\}, \{13, 21\}, \{14, 23\}, \{15, 31\}, \right.\\
           & \left. \{16, 28\}, \{19, 30\}, \{20, 26\},  \{27, 32\}, \{29, 33\} \right\}.\\
    \end{split}
\end{equation*}
These two $\mathbb{C}_{34}$ instances are not twin instances. In addition, even for a length 6 covered by the family of length $p-1$ presented in \cite{cyclic-code}, besides the instance $C_0=\left\{ \{1,3\},\{4,5\} \right\}$ constructed in \cite{cyclic-code}, we can find another instance in Table~\ref{table-example}. These two $\mathbb{C}_{6}$ instances are also not twin instances. Furthermore, as will be shown in the next section, there exist four families of $\mathbb{C}_{p-1}$ instances. Besides, Table~\ref{table-number-of-code} gives the number of \mbox{C-Codes} for even lengths from 4 to 30. These results are derived from \cite{P1F-list-(even)-starter-K32}. From this table, we can see that for most but not all of even lengths, the number of \mbox{C-Codes} increases with the length. We can also observe that the number of \mbox{C-Codes} for each length is always an even number. The reason is that each \mbox{C-Code} instance always has a twin instance (see Theorem~\ref{theorem-twin}).
\begin{table}[H]
%  \centering
  \caption{The number (\#) of \mbox{C-Codes} for even lengths from 4 to 30.}\label{table-number-of-code}
  \begin{tabular}{|c||c|c|c|c|c|c|c|}
    \hline
    Length & 4 & 6 & 8 & 10 & 12 & 14 & 16  \\
    \hline
    \# & 2 & 4 & 0 & 16 & 24 & 12 & 80  \\
    \hline
    \hline
    Length & 18 & 20 & 22 & 24 & 26 & 28 & 30 \\
    \hline
    \# & 120 & 272 & 440 & 576 & 2016 & 4992 & 11104 \\
    \hline
  \end{tabular}
\end{table}

\section{Four Infinite Families of $\mathbb{C}_{p-1}$ Instances}\label{section-c-code-family}
\indent In the previous section, we have obtained \mbox{C-Codes} for some sporadic even lengths listed as follows: \[4, 6, 10, 12, 14, 16, 18, 20, 22, 24, 26, 28, 30, 32, 34, 36, 38, 50.\] Then, we may wonder if there exists an infinite family of \mbox{C-Codes}. A positive answer to this question will be given in this section. Exactly speaking, this section will construct four infinite families of $\mathbb{C}_{p-1}$ instances.\\
\indent We first give the extended definition of an event starter as follows:
\begin{definition}[\cite{even-starter-first}]
Let $(A_{2n},\circ)$ be an abelian group (written multiplicatively) of order $2n$ with identity $e$ and unique element $a^*$ of order 2 (i.e. $a^*\circ a^*=e$). An \emph{even starter} $\widehat{S}_E$ in $(A_{2n},\circ)$ is a set of $n-1$ pairs of \mbox{non-identity} elements of $A_{2n}$, i.e. \[\widehat{S}_E=\left\{
\{x_1,y_1\},\{x_2,y_2\},\cdots,\{x_{n-1},y_{n-1}\} \right\},\] such that for every element $i\in A_{2n}$ such that $i\neq e$ and $i\neq a^*$, there exists a pair $\{x,y\}\in \widehat{S}_E$ such that $i=x^{-1}\circ y$, or $i=x\circ y^{-1}$.
\end{definition}

\indent Let $\left(Z^*_{p},\times\right)$ be a multiplicative group of congruence classes modulo $p$. For $p=7$, \[Z^*_{7}=\{1,2,\cdots,6\}.\] It is clear that $(Z^*_{7},\times)$ is an abelian group of order 6 with identity 1 and unique element 6 of order 2. Take $\widehat{S}_E=\left\{\{2,6\},\{3,5\}\right\}$ in $(Z^*_{7},\times)$ for example. For every \mbox{non-identity} element of $Z^*_{p}$ except 6, we have
\[
\left\{
  \begin{array}{ll}
    2=3\times 5^{-1} \bmod{7}; \\
    3=2^{-1}\times 6 \bmod{7}; \\
    4=3^{-1}\times 5 \bmod{7}; \\
    5=2\times 6^{-1} \bmod{7}. \\
  \end{array}
\right.
\]
Thus, $\widehat{S}_E$ is an even starter in $(Z^*_{7},\times)$.\\
\indent An even starter $\widehat{S}_E$ in $(A_{2n},\circ)$ induces a bipyramidal \mbox{one-factorization} of a \mbox{$2n$-regular} graph on $2n+2$ vertices as follows. Label the $2n+2$ vertices with the elements of $A_{2n}$ and two infinity elements $\infty_1$ and $\infty_2$ such that there is no edge between the following pairs of vertices: $\{\infty_1,\infty_2\}$ and all $\{a,a\circ a^*\}$ for $a\in A_{2n}$. Let
\begin{equation}\label{equation-expanded-even}
\overset{\sim}{S}_E=\widehat{S}_E\cup\left\{ \{e,\infty_1\},\{r,\infty_2\} \right\},
\end{equation}
where $r$ is the \mbox{non-identity} element that does not appear in $\widehat{S}_E$. For all $a\in A_{2n}$, define $a\circ\infty_1=\infty_1$ and $a\circ\infty_2=\infty_2$. The corresponding bipyramidal \mbox{one-factorization} $\mathbb{F}$ is then given by
\begin{equation}
\mathbb{F}=\left\{ a\circ \overset{\sim}{S}_E: a\in A_{2n} \right\},
\end{equation}
where \[a\circ \overset{\sim}{S}_E=\left\{\{a\circ x, a\circ y\}: \{x,y\}\in \overset{\sim}{S}_E \right\}.\] \\
\indent Here, if the bipyramidal \mbox{one-factorization} induced by $\widehat{S}_E$ in $(A_{2n},\circ)$ is a P1F, a \mbox{non-cyclic} \mbox{B-Code} of length $2n$, in which the \mbox{$a$-th} column ($a\in A_{2n}$) is $a\circ \widehat{S}_E$, can be constructed \cite{B-Code}.\\
\indent Take the foregoing even starter $\widehat{S}_E=\left\{\{2,6\},\{3,5\}\right\}$ in $(Z^*_{7},\times)$ for example. Since the corresponding bipyramidal \mbox{one-factorization} is a P1F, a \mbox{non-cyclic} \mbox{B-Code} of length 6 constructed using $\widehat{S}_E$ is illustrated as follows:
\begin{center}
  \begin{tabular}{|c|c|c|c|c|c|}
  \hline
  % after \\: \hline or \cline{col1-col2} \cline{col3-col4} ...
  $d_{2,6}$ & $d_{4,5}$ & $d_{6,4}$ & $d_{1,3}$ & $d_{3,2}$ & $d_{5,1}$ \\
  \hline
  $d_{3,5}$ & $d_{6,3}$ & $d_{2,1}$ & $d_{5,6}$ & $d_{1,4}$ & $d_{4,2}$ \\
  \hline
  \hline
  $p_1$ & $p_2$ & $p_3$ & $p_4$ & $p_5$ & $p_6$ \\
  \hline
  \end{tabular}
  .
\end{center}
This code is the same as the \mbox{P-Code} of length 6 constructed in \cite{P-Code} (in fact, \mbox{P-Code} is just one family of the \mbox{B-Code} of length $p-1$, and its code structure was originally derived from \cite{ZZS}).\\
\indent We now consider the case where $(A_{2n},\circ)$ is a cyclic group of which $g$ is a generator. Then, in the \mbox{non-cyclic} \mbox{B-Code} of length $2n$ constructed using $\widehat{S}_E$ in $(A_{2n},\circ)$, the \mbox{$(g^i)$-th} column ($i=0,1,\cdots,2n-1$) can be expressed as $g^i\circ \widehat{S}_E$. For $i=0,1,\cdots,2n-1$, replace $g^i$ with $i$ and then relabel the \mbox{$(g^i)$-th} column with $i$. Reorder all the $2n$ columns in order according to their new labels. Then, a \mbox{C-Code} of length $2n$ is obtained.\\
\indent Take the foregoing \mbox{non-cyclic} \mbox{B-Code} of length 6 for example. Since $(Z^*_{7},\times)$ is a cyclic group of which 3 is a generator, the code can then be represented by
\begin{center}
  \begin{tabular}{|c|c|c|c|c|c|}
  \hline
  % after \\: \hline or \cline{col1-col2} \cline{col3-col4} ...
  $d_{{3^2},{3^3}}$ & $d_{{3^4},{3^5}}$ & $d_{{3^3},{3^4}}$ & $d_{{3^0},{3^1}}$ & $d_{{3^1},{3^2}}$ & $d_{{3^5},{3^0}}$ \\
  \hline
  $d_{{3^1},{3^5}}$ & $d_{{3^3},{3^1}}$ & $d_{{3^2},{3^0}}$ & $d_{{3^5},{3^3}}$ & $d_{{3^0},{3^4}}$ & $d_{{3^4},{3^2}}$ \\
  \hline
  \hline
  $p_{3^0}$ & $p_{3^2}$ & $p_{3^1}$ & $p_{3^4}$ & $p_{3^5}$ & $p_{3^3}$ \\
  \hline
  \end{tabular}
  .
\end{center}
In the above representation, replace $3^i$ with $i$ for $i=0,1,\cdots,5$ and then reorder all the 6 columns in order according to their new labels. We can obtain a $\mathbb{C}_6$ instance as follows:
\begin{center}
  \begin{tabular}{|c|c|c|c|c|c|}
  \hline
  % after \\: \hline or \cline{col1-col2} \cline{col3-col4} ...
  $d_{2,3}$ & $d_{3,4}$ & $d_{4,5}$ & $d_{5,0}$ & $d_{0,1}$ & $d_{1,2}$ \\
  \hline
  $d_{1,5}$ & $d_{2,0}$ & $d_{3,1}$ & $d_{4,2}$ & $d_{5,3}$ & $d_{0,4}$ \\
  \hline
  \hline
  $p_0$ & $p_1$ & $p_2$ & $p_3$ & $p_4$ & $p_5$ \\
  \hline
  \end{tabular}
  .
\end{center}

\indent In a cyclic group (written multiplicatively) of which $g$ is a generator, for $x=g^i$, define $\log_g{(x)}=i$. Then, we can easily make the following conclusion:
\begin{theorem}
If $(A_{2n},\circ)$ is a cyclic group of which $g$ is a generator, a \mbox{non-cyclic} \mbox{B-Code} of length $2n$ constructed using an even starter $\widehat{S}_E$ in $(A_{2n},\circ)$ can always be transformed to a $\mathbb{C}_{2n}$ instance
\begin{equation}\label{equation-transform-c-code}
C_0=\left\{\left\{\log_g{(x)},\log_g{(y)}\right\}:\{x,y\}\in\widehat{S}_E\right\}.
\end{equation}
At the same time, we can get the twin $\mathbb{C}_{2n}$ instance
\begin{equation}\label{equation-twin-c-code}
C_0^\tau=\left\{\{x^*-r^*,y^*-r^*\}: \{x^*,y^*\}\in C_0 \right\},
\end{equation}
where $r^*$ is the one and only \mbox{non-zero} element of $(Z_{2n},+)$ that does not occur in $C_0$.
\end{theorem}

\indent Specially, we consider even starters in $\left(Z^*_{p},\times\right)$, where $p$ is a prime. It is \mbox{well-known} that when $p$ is a prime, $\left(Z^*_{p},\times\right)$ is a cyclic group in which \[Z^*_{p}=\{1,2,\cdots,p-1\}.\] Thus, we can make the following conclusion:
\begin{theorem}\label{theorem-even-example}
A \mbox{non-cyclic} \mbox{B-Code} of length $p-1$ constructed using an even starter in $\left(Z^*_{p},\times\right)$ can always be transformed to a \mbox{C-Code} of length $p-1$.
\end{theorem}

\indent Finally, in $\left(Z^*_{p},\times\right)$, there exist two infinite families of even starters \cite{2-even-starter-induced-P1Fs} as follows:
\begin{equation}\label{equation-even-starter1}
\widehat{S}_E^A=\left\{\{x,y\}: x,y\in Z^*_{p}\setminus\{1,2^{-1}\}, x+y=1\right\},
\end{equation}
and
\begin{equation}\label{equation-even-starter2}
\widehat{S}_E^B= \left\{\{x,y\}: x,y\in Z^*_{p}\setminus\{1,2^{-1},2,p-1\}, x+y=1\right\}\bigcup\left\{\{2^{-1},p-1\}\right\}.
\end{equation}
It was proved in \cite{2-even-starter-induced-P1Fs} that $\widehat{S}_E^A$ and $\widehat{S}_E^B$ can induce two families of \mbox{non-isomorphic} bipyramidal P1Fs of $K_{p+1}$, respectively. Note that the P1F induced by $\widehat{S}_E^A$ is isomorphic to the \mbox{well-known} \emph{patterned P1F} (induced by the \mbox{well-known} \emph{patterned starter} in $(Z_p,+)$) \cite{P1F-2-families}, which has been used to construct the family of \mbox{non-cyclic} \mbox{B-Codes} of length $p-1$ in \cite{P-Code}, \cite{B-Code}, and \cite{ZZS}. Thus, two families of \mbox{non-cyclic} \mbox{B-Codes} of length $p-1$ can be constructed using $\widehat{S}_E^A$ and $\widehat{S}_E^B$, respectively. \\
\indent Suppose $g$ is a generator of $\left(Z^*_{p},\times\right)$. Then, two families of $\mathbb{C}_{p-1}$ instances
\begin{equation}\label{equation-even-starter1-code}
C_0^A=\left\{\left\{\log_g{(x)},\log_g{(y)}\right\}: x,y\in Z^*_{p}\setminus\{1,2^{-1}\}, x+y=1\right\}
\end{equation}
and
\begin{equation}\label{equation-even-starter2-code}
C_0^B= \left\{\left\{\log_g{(x)},\log_g{(y)}\right\}: x,y\in Z^*_{p}\setminus\{1,2^{-1},2,p-1\}, x+y=1\right\}\bigcup\left\{\{2^{-1},p-1\}\right\}
\end{equation}
and their twin instances can be constructed. Therefore, there exist four families of $\mathbb{C}_{p-1}$ instances.\\
\indent Take $(Z^*_{7},\times)$ for example. We then have $\widehat{S}_E^A=\left\{\{2,6\},\{3,5\}\right\}$ and $\widehat{S}_E^A=\left\{\{3,5\},\{4,6\}\right\}$. In $\left(Z^*_{7},\times\right)$, pick $g=3$. From $\widehat{S}_E^A$, we obtain a $\mathbb{C}_{6}$ instance \[C_0^A=\left\{\{2,3\},\{1,5\}\right\}\] and its twin instance \[\left(C_0^A\right)^\tau=\left\{\{4,5\},\{3,1\}\right\}.\] Here, the twin instance is the same as the instance constructed in \cite{cyclic-code}. Also, from $\widehat{S}_E^B$, we obtain a $\mathbb{C}_{6}$ instance \[C_0^B=\left\{\{1,5\},\{4,3\}\right\}\] and its twin instance \[\left(C_0^B\right)^\tau=\left\{\{5,3\},\{2,1\}\right\}.\] Thus, we construct four $\mathbb{C}_{6}$ instances.\\
\indent From the above results, we can make two observations as follows:
\begin{enumerate}
  \item \mbox{Non-cyclic} \mbox{B-Codes} of length $p-1$ constructed in \cite{P-Code}, \cite{B-Code}, and \cite{ZZS} can always be transformed to $\mathbb{C}_{p-1}$ instances; and
  \item The family of $\mathbb{C}_{p-1}$ instances constructed in \cite{cyclic-code} can also be obtained from $\widehat{S}_E^A$.
\end{enumerate}

\section{Constructing a \mbox{Quasi-C-Code} Using Even \mbox{Multi-Starters}}\label{section-quasi-c-code}
As mentioned in Section~\ref{section-c-code}, our exhaustive search showed that there is no \mbox{C-Code} for some even lengths, such as 8. Then, one question is: \emph{Can we construct \mbox{quasi-C-Codes} (which partially hold cyclic symmetry \cite{cyclic-code}) for these even lengths?} In this paper, we say a \mbox{quasi-C-Code} of length $2n$ is a \emph{\mbox{$\kappa$-quasi-C-Code}} (denoted by $\mathbb{C}_{2n}^\kappa$, where $\kappa|2n$) if for $i=0,1,\cdots,\kappa-1$, each group of $\frac{2n}{\kappa}$ columns \[C_{i+\kappa\times 0},C_{i+\kappa\times 1},\cdots,C_{i+\kappa\times(\frac{2n}{\kappa}-1)}\] hold cyclic symmetry, where $C_j$ represents the $j$-th column of the code.\\
\indent In this section, we will introduce a concept of even \mbox{multi-starters} and then discuss how to construct \mbox{quasi-C-Codes} using even \mbox{multi-starters}.\\
\indent An even \mbox{$\kappa$-starter} in $(Z_{2n},+)$ (where $\kappa|2n$) is defined as follows:

\begin{definition}\label{definition-multistarter}
An even \emph{\mbox{$\kappa$-starter}} $S^{\kappa}$ in $(Z_{2n},+)$ (where $\kappa|2n$) is a set \[S^{\kappa}=\left\{ S_0, S_1, \cdots, S_{\kappa-1} \right\},\] where $S_i$ ($i=0,1,\cdots,\kappa-1$) is a set of $n-1$ pairs of \mbox{non-$i$} elements of $Z_{2n}$, such that every integer from 1 to $n-1$ occurs $\kappa$ times as a difference of a pair of $S^{\kappa}$. Its \emph{twin even \mbox{$\kappa$-starter}} ${(S^{\kappa})}^\tau$ is defined as
\begin{equation}\label{equation-twin-even-multi-starter}
    {(S^{\kappa})}^\tau=\left\{ S'_{r_i \bmod{\kappa}}=S_i-\kappa\left\lfloor\frac{r_i}{\kappa}\right\rfloor: i=0,1,\cdots,\kappa-1 \right\},
\end{equation}
where $r_i$ is the \mbox{non-$i$} element that does not appear in $S_i$, and \[S_i-\kappa\left\lfloor\frac{r_i}{\kappa}\right\rfloor=\left\{\left\{x-\kappa\left\lfloor\frac{r_i}{\kappa}\right\rfloor,y-\kappa\left\lfloor\frac{r_i}{\kappa}\right\rfloor\right\}:\{x,y\}\in S_i\right\}.\]
\end{definition}

\indent Take $S^2=\left\{ S_0, S_1\right\}$ in $(Z_8,+)$ for example, where $S_0=\left\{ \{1,2\}, \{3,5\}, \{4,6\} \right\}$, and $S_1=\left\{ \{0,3\}, \{2,7\}, \{4,5\} \right\}$. For every integer from 1 to 3, we have
\[
\left\{
  \begin{array}{ll}
    1=2-1=5-4 \bmod{8}; \\
    2=5-3=6-4 \bmod{8}; \\
    3=3-0=2-7 \bmod{8}. \\
  \end{array}
\right.
\]
Thus, $S^2$ is an even \mbox{2-starter} in $(Z_8,+)$. Its twin even \mbox{2-starter} is ${(S^2)}^\tau=\left\{ S'_0, S'_1\right\}$, where $S'_0=\left\{ \{2,5\}, \{4,1\}, \{6,7\} \right\}$, and $S'_1=\left\{ \{3,4\}, \{5,7\}, \{6,0\} \right\}$. \\
\indent An even \mbox{$\kappa$-starter} \[S^{\kappa}=\left\{ S_0, S_1, \cdots, S_{\kappa-1} \right\}\] in $(Z_{2n},+)$ (where $\kappa|2n$) induces a \mbox{one-factorization} of a \mbox{$2n$-regular} graph on $2n+2$ vertices as follows. Label these $2n+2$ vertices with the elements of $Z_{2n}$ and two infinity elements $\infty_1$ and $\infty_2$ such that there is no edge between the following pairs of vertices: $\{\infty_1,\infty_2\}$
and all $\{i,i+n\}$ for $i=0,1,\cdots,n-1$. For every $z\in Z_{2n}$, define $z+\infty_1=\infty_1$ and $z+\infty_2=\infty_2$. For
$i=0,1,\cdots,\kappa-1$, let
\begin{equation}
\overset{\sim}{S}_i=S_i\cup\left\{ \{i,\infty_1\}, \{r_i,\infty_2\} \right\},
\end{equation}
where $r_i$ is the \mbox{non-$i$} element that does not appear in $S_i$. The corresponding \mbox{one-factorization} $\mathbb{F}^{\kappa}$ is then given by
\begin{equation}
\mathbb{F}^{\kappa}= \left\{\kappa\tau+\overset{\sim}{S}_0,\kappa\tau+\overset{\sim}{S}_1,\cdots,\kappa\tau+\overset{\sim}{S}_{\kappa-1}: \tau=0,1,\cdots,\frac{2n}{\kappa}-1\right\},
\end{equation}
where \[\kappa\tau+\overset{\sim}{S}_i=\left\{ \{\kappa\tau+x,\kappa\tau+y\}:\{x,y\}\in \overset{\sim}{S}_i \right\}\] for $i=0,1,\cdots,\kappa-1$. Such a \mbox{one-factorization} is called a \emph{\mbox{$\kappa$-quasi-bipyramidal}  \mbox{one-factorization}}.\\
\indent For an even \mbox{$\kappa$-starter} \[S^{\kappa}=\left\{ S_0, S_1, \cdots, S_{\kappa-1} \right\}\] in $(Z_{2n},+)$ (where $\kappa|2n$), if the \mbox{$\kappa$-quasi-bipyramidal} \mbox{one-factorization} $\mathbb{F}^{\kappa}$ induced by $S^{\kappa}$ is a P1F of a \mbox{$2n$-regular} graph on $2n+2$ vertices, a $\mathbb{C}_{2n}^\kappa$ instance, in which the \mbox{$i$-th} column
($i=0,1,\cdots,2n-1$) is
\begin{equation}
C_i=  \left\{ \left\{x+\kappa\left\lfloor\frac{i}{\kappa}\right\rfloor \bmod{2n},y+\kappa\left\lfloor\frac{i}{\kappa}\right\rfloor \bmod{2n}\right\}: \{x,y\}\in S_{i \bmod{\kappa}}\right\},
\end{equation}
can be constructed using $S^{\kappa}$. It can be easily checked that in this $\mathbb{C}_{2n}^\kappa$ instance, for $i=0,1,\cdots,\kappa-1$, each group of $\frac{2n}{\kappa}$ columns \[C_{i+\kappa\times 0},C_{i+\kappa\times 1},\cdots,C_{i+\kappa\times(\frac{2n}{\kappa}-1)}\] hold cyclic symmetry. \\
\indent Similar to Theorem~\ref{theorem-twin} in Section~\ref{section-c-code}, we give the following theorem:
\begin{theorem}\label{theorem-twin-quasi}
If a $\mathbb{C}_{2n}^\kappa$ instance can be constructed using an even \mbox{$\kappa$-starter} $S^{\kappa}$ in $(Z_{2n},+)$, another $\mathbb{C}_{2n}^\kappa$ instance can also be constructed using the twin even \mbox{$\kappa$-starter} $(S^{\kappa})^\tau$. They are called \emph{twin $\mathbb{C}_{2n}^\kappa$ instances}.
\end{theorem}

\indent The above conclusion can also be easily understood because twin even \mbox{$\kappa$-starters} $S^{\kappa}$ and $(S^{\kappa})^\tau$ induce the same \mbox{$\kappa$-quasi-bipyramidal} one-factorization.\\
\indent For example, the foregoing even \mbox{2-starter} $S^2=\left\{ S_0, S_1\right\}$ (where $S_0=\left\{ \{1,2\}, \{3,5\}, \{4,6\} \right\}$, and $S_1=\left\{ \{0,3\}, \{2,7\}, \{4,5\} \right\}$) in $(Z_8,+)$ induces a \mbox{2-quasi-bipyramidal} P1F of a \mbox{8-regular} graph on 10 vertices. Thus, a $\mathbb{C}_{8}^2$ instance illustrated as follows can be constructed using $S^2$:
\begin{center}
  \begin{tabular}{|c|c|c|c|c|c|c|c|}
  \hline
  % after \\: \hline or \cline{col1-col2} \cline{col3-col4} ...
  $d_{1,2}$ & $d_{0,3}$ & $d_{3,4}$ & $d_{2,5}$ & $d_{5,6}$ & $d_{4,7}$ & $d_{7,0}$ & $d_{6,1}$ \\
  \hline
  $d_{3,5}$ & $d_{2,7}$ & $d_{5,7}$ & $d_{4,1}$ & $d_{7,1}$ & $d_{6,3}$ & $d_{1,3}$ & $d_{0,5}$ \\
  \hline
  $d_{4,6}$ & $d_{4,5}$ & $d_{6,0}$ & $d_{6,7}$ & $d_{0,2}$ & $d_{0,1}$ & $d_{2,4}$ & $d_{2,3}$ \\
  \hline
  \hline
  $p_0$ & $p_1$ & $p_2$ & $p_3$ & $p_4$ & $p_5$ & $p_6$ & $p_7$ \\
  \hline
  \end{tabular}
  .
\end{center}
At the same time, we can construct the twin $\mathbb{C}_{8}^2$ instance illustrated as follows using the twin even \mbox{2-starter} ${(S^2)}^\tau=\left\{ S'_0, S'_1\right\}$, where $S'_0=\left\{ \{2,5\}, \{4,1\}, \{6,7\} \right\}$, and $S'_1=\left\{ \{3,4\}, \{5,7\}, \{6,0\} \right\}$:
\begin{center}
  \begin{tabular}{|c|c|c|c|c|c|c|c|}
  \hline
  % after \\: \hline or \cline{col1-col2} \cline{col3-col4} ...
  $d_{2,5}$ & $d_{3,4}$ & $d_{4,7}$ & $d_{5,6}$ & $d_{6,1}$ & $d_{7,0}$ & $d_{0,3}$ & $d_{1,2}$ \\
  \hline
  $d_{4,1}$ & $d_{5,7}$ & $d_{6,3}$ & $d_{7,1}$ & $d_{0,5}$ & $d_{1,3}$ & $d_{2,7}$ & $d_{3,5}$ \\
  \hline
  $d_{6,7}$ & $d_{6,0}$ & $d_{0,1}$ & $d_{0,2}$ & $d_{2,3}$ & $d_{2,4}$ & $d_{4,5}$ & $d_{4,6}$ \\
  \hline
  \hline
  $p_0$ & $p_1$ & $p_2$ & $p_3$ & $p_4$ & $p_5$ & $p_6$ & $p_7$ \\
  \hline
  \end{tabular}
  .
\end{center}

\section{Two Infinite Families of $\mathbb{C}_{2(p-1)}^2$ Instances}\label{section-c-code-family-quasi}
\indent In this section, we will construct two infinite families of $\mathbb{C}_{2(p-1)}^2$ instances. We start with an infinite family of even \mbox{2-starter} $S^2$ and its twin even \mbox{2-starter} $\left(S^2\right)^\tau$ in $\left(Z_{2(p-1)},+\right)$. \\
\indent Suppose $g$ is a generator of $\left(Z^*_{p},\times\right)$. In $\left(Z^*_{p},\times\right)$, for $x=g^i$, define $\log_g{(x)}=i$. Then, $S^2=\left\{ S_0, S_1 \right\}$ in $\left(Z_{2(p-1)},+\right)$ is constructed as follows:
\begin{equation}
S_0= \left\{\left\{2\log_g{(x)},2\log_g{(y)}+1\right\}: x\in Z_p^*\setminus\{1\}, y\in Z_p^*\setminus\{p-1\},x-y=1\right\},
\end{equation}
and
\begin{equation}
S_1= \left\{\{2x+1,2y+1\}:\{x,y\}\in C_0^A\right\}\bigcup \left\{\{2x,2y\}:\{x,y\}\in C_0^A\right\}\bigcup \left\{\{2r,2r+1\}\right\},
\end{equation}
where $C_0^A$ is defined in Equation~(\ref{equation-even-starter1-code}) in Section~\ref{section-c-code-family}, and $r$ is the one and only \mbox{non-zero} element of $\left(Z_{p-1},+\right)$ that does not occur in $C_0^A$.
Its twin even \mbox{2-starter} in $\left(Z_{2(p-1)},+\right)$ is $\left(S^2\right)^\tau=\left\{ S'_0, S'_1 \right\}$, where
\begin{equation}
S'_0=S_1,
\end{equation}
and
\begin{equation}
S'_1= \left\{\left\{2\log_g{(x)}+1,2\log_g{(y)}\right\}: x\in Z_p^*\setminus\{1\}, y\in Z_p^*\setminus\{p-1\},x-y=1\right\}.
\end{equation}

\indent Take $\left(Z_{8},+\right)$ for example. Pick $g=2$ in $\left(Z^*_{5},\times\right)$. We then have $S^2=\left\{ S_0, S_1 \right\}$, where $S_0=\left\{ \{2,1\}, \{6,3\}, \{4,7\} \right\}$, and $S_1=\left\{ \{2,4\}, \{3,5\}, \{6,7\} \right\}$. Its twin even \mbox{2-starter} in $\left(Z_{8},+\right)$ is $\left(S^2\right)^\tau=\left\{ S'_0, S'_1 \right\}$, where $S'_0=\left\{ \{2,4\}, \{3,5\}, \{6,7\} \right\}$, and $S'_1=\left\{ \{3,0\}, \{7,2\}, \{5,6\} \right\}$.\\
\indent It can be verified that this family of even \mbox{2-starter} $S^2$ and its twin even \mbox{2-starter} $\left(S^2\right)^\tau$ in $\left(Z_{2(p-1)},+\right)$ can induce the same \mbox{2-quasi-bipyramidal} P1F of $K_{2p}$, which is isomorphic to the \mbox{well-known} P1F $GA_{2p}$ of $K_{2p}$ \cite{P1F-2-families}. Thus, two families of $\mathbb{C}_{2(p-1)}^2$ instances can be constructed using $S^2$ and $\left(S^2\right)^\tau$, respectively.\\
\indent Here, the family of $\mathbb{C}_{2(p-1)}^2$ instances constructed using $S^2$ can be shown to be the same as those constructed in \cite{cyclic-code}. Besides, it was proved in \cite{P1F-2-k2p} that the P1F $GN_{2p}$ of $K_{2p}$, which was adopted in \cite{B-Code} to construct the family of \mbox{non-cyclic} \mbox{B-Codes} of length $2(p-1)$, is also isomorphic to $GA_{2p}$. Thus, we can make two observations as follows:
\begin{enumerate}
  \item \mbox{Non-cyclic} \mbox{B-Codes} of length $2(p-1)$ constructed in \cite{B-Code} can always be transformed to $\mathbb{C}_{2(p-1)}^2$ instances; and
  \item The family of $\mathbb{C}_{2(p-1)}^2$ instances constructed in \cite{cyclic-code} can also be constructed using $S^2$.
\end{enumerate}
\section{Conclusions and Remarks}
This paper investigated the underlying connections between \mbox{distance-3} cyclic (or \mbox{quasi-cyclic}) \mbox{lowest-density} MDS array codes and starters in group theory. Some interesting new results listed as follows were obtained:
\begin{enumerate}
  \item Each cyclic code of length $2n$ can be constructed using an even starter in $(Z_{2n},+)$ (see Section~\ref{section-c-code}), while each \mbox{quasi-cyclic} code of length $2n$ can be constructed using an even \mbox{multi-starter} in $(Z_{2n},+)$ (see Section~\ref{section-quasi-c-code});
  \item Each cyclic (or \mbox{quasi-cyclic}) code has a twin cyclic (or \mbox{quasi-cyclic}) code (see Sections~\ref{section-c-code}~and~\ref{section-quasi-c-code});
  \item A cyclic code exists for most but not all of even lengths (one exception is 8) (see Section~\ref{section-c-code});
  \item Four infinite families of cyclic codes of length $p-1$ (which cover the family of cyclic codes of length $p-1$ constructed in \cite{cyclic-code}) were constructed from two infinite families of even starters in $\left(Z^*_{p},\times\right)$ (where $p$ is a prime) in Section~\ref{section-c-code-family};
  \item Besides the family of length $p-1$, cyclic codes for some sporadic even lengths listed as follows were obtained in Section~\ref{section-c-code}: \[14, 20, 24, 26, 32, 34, 38, 50;\]
  \item Two infinite families of \mbox{quasi-cyclic} codes of length $2(p-1)$ (which cover the family of \mbox{quasi-cyclic} codes of length $2(p-1)$ constructed in \cite{cyclic-code}) were constructed using an infinite family of even \mbox{2-starter} in $\left(Z_{2(p-1)},+\right)$ in Section~\ref{section-c-code-family-quasi}; and
  \item \mbox{Non-cyclic} \mbox{B-Codes} of length $p-1$ constructed in \cite{P-Code}, \cite{B-Code}, and \cite{ZZS} can always be transformed to cyclic codes (see Section~\ref{section-c-code-family}), while  \mbox{non-cyclic} \mbox{B-Codes} of length $2(p-1)$ constructed in \cite{B-Code} can always be transformed to \mbox{quasi-cyclic} codes (see Section~\ref{section-c-code-family-quasi}).
\end{enumerate}

\begin{table}[H]
%  \centering
  \caption{The existence of \mbox{distance-3} cyclic (or \mbox{quasi-cyclic}) \mbox{lowest-density} MDS array codes for even lengths from 4 to 58.}\label{table-result-total}
  \begin{minipage}[t]{1\textwidth}
  \begin{center}
  \begin{tabular}{|c||c|c|c|c|c|c|c|c|c|c|c|c|c|c|}
    \hline
    Length & 4 & 6 & 8 & 10 & 12 & 14 & 16 & 18 & 20 & 22 & 24 & 26 & 28 & 30 \\
    \hline
    Cyclic  & $\surd$  & $\surd$  & $\times$ & $\surd$ & $\surd$ & $\surd$ & $\surd$  & $\surd$ & $\surd$ & $\surd$ & $\surd$ & $\surd$ & $\surd$ & $\surd$ \\
    \hline
    \mbox{Quasi-Cyclic}  &    &    &   $\surd$ &  &   $\surd$  &    &   &    &  $\surd$  &   & $\surd$  &   &   &  \\
    \hline
    \hline
    Length & 32 & 34 & 36 & 38 & 40 & 42 & 44 & 46 & 48 & 50 & 52 & 54 & 56 & 58\\
    \hline
    Cyclic   & $\surd$ & $\surd$ & $\surd$ & $\surd$ & $\surd$ & $\surd$ & ?  &  $\surd$  &  ?  & $\surd$  &  $\surd$ &  ? & ?  &  $\surd$ \\
    \hline
    \mbox{Quasi-Cyclic}   &  $\surd$  &    &  $\surd$ &   &   &   &  $\surd$  &    &    &   &   &   &  $\surd$ &  \\
    \hline
  \end{tabular}\\[2pt]
  \end{center}
  \footnotesize $\surd$: existence; $\times$: inexistence; ?: unknown.
  \end{minipage}
\end{table}

\indent According to the above results, we can obtain Table~\ref{table-result-total}. From this table, we can see that for even lengths from 4 to 58, there are one length 8, for which the cyclic code does not exist, and four lengths 44, 48, 54, and 56, for which cyclic codes are still unknown. Luckily, \mbox{quasi-cyclic} codes for lengths 8, 44, and 56 can be constructed in Section~\ref{section-c-code-family-quasi}. Then, the constructions of cyclic (or \mbox{quasi-cyclic}) codes for the rest two lengths 48 and 54 are left as open problems. Here, two points deserve future researchers' attention:
\begin{enumerate}
  \item \mbox{Non-cyclic} \mbox{B-Codes} of length 48 can be constructed using P1Fs of $K_{50}$ found in \cite{P1F-K50-starter}. However, these P1Fs were induced by starters in $(Z_{49},+)$. Whether these \mbox{non-cyclic} \mbox{B-Codes} can be transformed to cyclic codes is left as an open problem.
  \item Since 2009, when a P1F of $K_{52}$ was found (see \cite{P1F-K52-even-starter}), $K_{56}$ has been the smallest complete graph for which a P1F has not been known. The construction of a P1F of $K_{56}$ is left as an open problem in graph theory. Consequently, the construction of a \mbox{B-Code} of length 54 is still unknown.
\end{enumerate}

%% The Appendices part is started with the command \appendix;
%% appendix sections are then done as normal sections

\appendix
\section{Proof of Theorem~\ref{theorem-equivalence}}\label{appendix1}
We prove this theorem by two algorithms. Here, note that the basic idea comes from the work of \cite{B-Code}, and a similar proof was given in \cite{equivalence}.\\
\indent We now first propose Algorithm~\ref{alg1} to construct a $\mathbb{C}_{2n}$ instance from a bipyramidal P1F of a \mbox{$2n$-regular} graph on $2n+2$ vertices. In this algorithm, it is clear that $\left\{C_0,C_1,\cdots,C_{2n-1}\right\}$ meets Equation~(\ref{equation-cyclic}) in Section~\ref{subsection-description}. It can also be proved as follows that $\left\{C_0,C_1,\cdots,C_{2n-1}\right\}$ meets Condition~\ref{condition} in Section~\ref{subsection-description}. Thus, a corresponding $\mathbb{C}_{2n}$ instance is constructed by Algorithm~\ref{alg1}.
\begin{algorithm}[H]
\caption{Constructing a $\mathbb{C}_{2n}$ instance from a bipyramidal P1F of a \mbox{$2n$-regular} graph on $2n+2$ vertices.}
\label{alg1}
\begin{description}
    \item[(S1)] Choose arbitrary pair of vertices that are not adjacent to each other in the regular graph and label them with $\infty_1$ and $\infty_2$. Then, label the other $2n$ vertices of the regular graph with integers from 0 to $2n-1$.
    \item[(S2)] If a bipyramidal P1F exists for the regular graph, then let $F_i$ denote the \mbox{one-factor} that contains the edge $\{i,\infty_1\}$, where $i=0,1,\ldots,2n-1$.
    \item[(S3)] In each $F_i$, delete the two edges that are incident to the two vertices $\infty_1$ and $\infty_2$. Then, delete the two vertices $\infty_1$ and $\infty_2$ in the graph. For $i=0,1,\ldots,2n-1$, let $C_i=F_i\setminus \left\{ \{i,\infty_1\}, \{r_i,\infty_2\} \right\}$, where $r_i$ is the vertex that is adjacent to the vertex $\infty_2$ originally in $F_i$, and label all the edges in $C_i$ with $i$.
\end{description}
\end{algorithm}

\indent According to Definition~\ref{definition-P1F} in Section~\ref{section-introduction}, in a P1F, for any pair of \mbox{one-factors} $F_{i_1}$ and $F_{i_2}$, the union of them forms a Hamiltonian cycle. Then, in the union of $F_{i_1}$ and $F_{i_2}$, after we delete all the edges that are incident to the two vertices $\infty_1$ and $\infty_2$, no cycle can exist. In addition, there also does not exist a path whose terminal vertices are the two vertices $i_1$ and $i_2$, otherwise the union of the path and the two edges $\{i_1,\infty_1\}$ (contained in $F_{i_1}$) and $\{i_2,\infty_1\}$ (contained in $F_{i_2}$) can form a cycle that does not visit the vertex $\infty_2$, which conflicts with the fact that the union of the two \mbox{one-factors} $F_{i_1}$ and $F_{i_2}$ forms a Hamiltonian cycle. Thus, $\left\{C_0,C_1,\cdots,C_{2n-1}\right\}$ in Algorithm~\ref{alg1} meets Condition~\ref{condition} in Section~\ref{subsection-description}. \\
\indent To make Algorithm~\ref{alg1} more easily understood, we give an example of constructing the $\mathbb{C}_{4}$ instance in Section~\ref{subsection-definition} from a bipyramidal P1F of a \mbox{4-regular} graph on 6 vertices in Figure~\ref{fig-construction}. \\
\indent Then, the next natural question is: Can we get a bipyramidal P1F of a \mbox{$2n$-regular} graph on a set of $2n+2$ vertices from a $\mathbb{C}_{2n}$ instance? A positive answer to this question will be given by Algorithm~\ref{alg2}.
\begin{algorithm}[H]
\caption{Constructing a bipyramidal P1F of a \mbox{$2n$-regular} graph on $2n+2$ vertices from a $\mathbb{C}_{2n}$ instance.}
\label{alg2}
\begin{description}
    \item[(S1)] If a $\mathbb{C}_{2n}$ instance exists, describe the code using the graph representation mentioned in Section~\ref{subsection-description} and let $C_i$ represent the \mbox{$i$-th} column of the code in the graph description, where $i=0,1,\ldots,2n-1$.
    \item[(S2)] Add two vertices $\infty_1$ and $\infty_2$ to the \mbox{$(2n-2)$-regular} graph $G$ of vertices $0,1,\ldots,2n-1$.
    \item[(S3)] For $i=0,1,\ldots,2n-1$, add two edges $\{i,\infty_1\}$ and $\{r_i,\infty_2\}$ to $C_i$, where $r_i$ is an integer from 0 to $2n-1$ such that the expanded set $\overset{\sim}{C}_i$ is a \mbox{one-factor} of the expanded graph $\overset{\sim}{G}$ of vertices $0,1,\cdots,2n-1,\infty_1,\infty_2$.
\end{description}
\end{algorithm}

\indent In Algorithm~\ref{alg2}, for $i=0,1,\cdots,2n-1$, $\overset{\sim}{C}_i$ has the following form:
\begin{equation}
    \overset{\sim}{C}_i=C_i\cup\left\{ \{i,\infty_1\}, \{r_i,\infty_2\} \right\}.
\end{equation}

\indent It is clear that the new graph $\overset{\sim}{G}$ is a \mbox{$2n$-regular} graph on $2n+2$ vertices.\\
\indent For $i=0,1,\cdots,2n-1$, define $\infty_1+i=\infty_1$ and $\infty_2+i=\infty_2$. Then, for $i=0,1,\cdots,2n-1$, we have
\begin{equation}\label{equation-expanded-cyclic}
    \overset{\sim}{C}_{i}=\left\{ \{x+i\bmod{2n},y+i\bmod{2n}\}:\{x,y\}\in \overset{\sim}{C}_0\right\}.
\end{equation}

\indent Consequently,
\begin{equation}\label{equation-P1F}
\mathbb{F}=\left\{ \overset{\sim}{C}_0,\overset{\sim}{C}_1,\cdots,\overset{\sim}{C}_{2n-1} \right\}
\end{equation}
is a bipyramidal \mbox{one-factorization} of a \mbox{$2n$-regular} graph on $2n+2$ vertices.\\
\indent Take the $\mathbb{C}_4$ instance in Section~\ref{subsection-definition} for example. Figure~\ref{fig-expanded-C4-graph} shows the corresponding expanded graph $\overset{\sim}{G}$, which is a \mbox{4-regular} graph on a set of 6 vertices $\{0,1,2,3,\infty_1,\infty_2\}$. The four corresponding expanded sets are
\[
\left\{
  \begin{array}{l}
    \overset{\sim}{C}_0=\left\{ \{1,2\}, \{0,\infty_1\}, \{3,\infty_2\} \right\}; \\
    \overset{\sim}{C}_1=\left\{ \{2,3\}, \{1,\infty_1\}, \{0,\infty_2\} \right\}; \\
    \overset{\sim}{C}_2=\left\{ \{3,0\}, \{2,\infty_1\}, \{1,\infty_2\} \right\}; \\
    \overset{\sim}{C}_3=\left\{ \{0,1\}, \{3,\infty_1\}, \{2,\infty_2\} \right\}.
  \end{array}
\right.
\]
It is clear that \[\mathbb{F}=\left\{ \overset{\sim}{C}_0,\overset{\sim}{C}_1,\overset{\sim}{C}_2,\overset{\sim}{C}_{3} \right\}\] is a bipyramidal \mbox{one-factorization} of a \mbox{4-regular} graph on 6 vertices.\\
\indent We then prove that the bipyramidal \mbox{one-factorization} obtained in Algorithm~\ref{alg2} is perfect as follows:\\
\indent From Condition~\ref{condition} in Section~\ref{subsection-description}, we can deduce that for any $m$ and $k$ (where $0\leq m<k\leq 2n-1$), the subgraph \[G^*=(\{0,1,\cdots,2n-1\},C_m\cup C_k)\] can be in one of the following two forms:

\begin{enumerate}
    \item $G^*$ consists of an isolated vertex $m$ (or $k$) and a path of length $2n-2$ one of whose terminal vertices is the other vertex $k$ (or $m$); or
    \item $G^*$ consists of two paths that satisfy: \rmnum{1}) the sum of their length is $2n-2$, and \rmnum{2}) one of the terminal vertices of each path is the vertex $m$ or $k$.
\end{enumerate}

\indent Then, in the \mbox{one-factorization} constructed in Algorithm~\ref{alg2}, the union of any pair of \mbox{one-factors} forms a Hamiltonian cycle. Thus, according to Definition~\ref{definition-P1F} in Section~\ref{section-introduction}, the \mbox{one-factorization} obtained in Algorithm~\ref{alg2} is a bipyramidal P1F of a \mbox{$2n$-regular} graph on $2n+2$ vertices.

\section{Proof of Theorem~\ref{theorem-bipyramidal-1F-even-starter}}\label{appendix2}
It is clear that $S$ consists of $n-1$ pairs of \mbox{non-zero} elements in $(Z_{2n},+)$. We then prove by contradiction that for every element $i\in Z_{2n}$ such that $i\neq 0$ and $i\neq n$, there exists a pair $\{x,y\}\in S_E$ such that $i=(x-y) \bmod{2n}$, or $i=(y-x) \bmod{2n}$.\\
\indent We first consider the first opposite case where for $i=n$, there exists a pair $\{x,y\}\in S_E$ such that $i=(x-y) \bmod{2n}$, or $i=(y-x) \bmod{2n}$. Suppose the corresponding pair is $\{x^*,y^*\}$, i.e. \[y^*-x^*=x^*-y^*=n.\] Then, we have \[\{x^*+n,y^*+n\}=\{x^*,y^*\}.\] $\mathbb{F}$ is a bipyramidal \mbox{one-factorization} in which \[F_n=\left\{ \{x+n,y+n\}:\{x,y\}\in F_0\right\}.\] Consequently, $\{x^*,y^*\}$ is contained in both $F_0$ and $F_n$ --- a contradiction!\\
\indent Now, under the condition $i\neq n$, we then consider the second opposite case where there is a \mbox{non-zero} element $i\in Z_{2n}$ such that there does not exist a pair $\{x,y\}\in S_E$ such that $i=(x-y) \bmod{2n}$, or $i=(y-x) \bmod{2n}$. Then, according to the pigeonhole principle, there exist two pairs $\{{x'}_1,{y'}_1\}$ and $\{{x'}_2,{y'}_2\}$, which meet \[{y'}_1-{x'}_1={y'}_2-{x'}_2.\] Then, we have \[{x'}_2-{x'}_1={y'}_2-{y'}_1.\] Let \[k={x'}_2-{x'}_1.\] Then, we have \[\{{x'}_1+k,{y'}_1+k\}=\{{x'}_2,{y'}_2\}.\] $\mathbb{F}$ is a bipyramidal \mbox{one-factorization} in which \[F_k=\left\{ \{x+k,y+k\}:\{x,y\}\in F_0\right\}.\] Consequently, $\{{x'}_2,{y'}_2\}$ is contained in both $F_0$ and $F_k$ --- a contradiction!\\
\indent Therefore, according to Definition~\ref{definition-evenstarter} in Section~\ref{section-c-code}, $S$ is an even starter in $(Z_{2n},+)$.

\bibliographystyle{spmpsci}
\bibliography{C-Code_Li}

\begin{thebibliography}{10}
\providecommand{\url}[1]{{#1}}
\providecommand{\urlprefix}{URL }
\expandafter\ifx\csname urlstyle\endcsname\relax
  \providecommand{\doi}[1]{DOI~\discretionary{}{}{}#1}\else
  \providecommand{\doi}{DOI~\discretionary{}{}{}\begingroup
  \urlstyle{rm}\Url}\fi

\bibitem{even-starter-first}
Anderson, B.A.: Sequencings and starters.
\newblock Pacific Journal of Mathematics \textbf{64}(1), 17--24 (1976)

\bibitem{P1F-2-families}
Anderson, B.A.: Symmetry groups of some perfect \mbox{1-factorizations} of
  complete graphs.
\newblock Discrete Mathematics \textbf{18}(3), 227--234 (1977)

\bibitem{latent_sector_error}
Bairavasundaram, L.N., Goodson, G.R., Pasupathy, S., Schindler, J.: An analysis
  of latent sector errors in disk drives.
\newblock In: Proceedings of the 2007 {ACM} {SIGMETRICS} International
  Conference on Measurement and Modeling of Computer Systems ({SIGMETRICS}'07),
  pp. 289--300. San Diego, CA (2007)

\bibitem{EVENODD}
Blaum, M., Brady, J., Bruck, J., Menon, J.: {EVENODD}: An efficient scheme for
  tolerating double disk failures in {RAID} architectures.
\newblock IEEE Transactions on Computers \textbf{44}(2), 192--202 (1995)

\bibitem{array-code}
Blaum, M., Farrell, P., van Tilborg, H.: Array codes.
\newblock In: V.~Pless, W.~Huffman (eds.) Handbook of Coding Theory, pp.
  1805--1909. Elsevier Science B.V., Amsterdam, The Netherlands (1998)

\bibitem{2-even-starter-induced-P1Fs}
Bryant, D., Maenhaut, B., Wanless, I.M.: New families of atomic {Latin} squares
  and perfect \mbox{1-factorisations}.
\newblock Journal of Combinatorial Theory, Series A \textbf{113}(4), 608--624
  (2006)

\bibitem{cyclic-code}
Cassuto, Y., Bruck, J.: Cyclic lowest density {MDS} array codes.
\newblock IEEE Transactions on Information Theory \textbf{55}(4), 1721--1729
  (2009)

\bibitem{RAID}
Chen, P.M., Lee, E.K., Gibson, G.A., Katz, R.H., Patterson, D.A.: {RAID}:
  High-performance, reliable secondary storage.
\newblock ACM Computing Surveys \textbf{26}(2), 145--185 (1994)

\bibitem{RDP}
Corbett, P., English, B., Goel, A., Grcanac, T., Kleiman, S., Leong, J.,
  Sankar, S.: Row-diagonal parity for double disk failure correction.
\newblock In: Proceedings of the 3rd {USENIX} Conference on File and Storage
  Technologies ({FAST}'04), pp. 1--14. San Francisco, CA (2004)

\bibitem{starter}
Dinitz, J.H.: Starters.
\newblock In: C.J. Colbourn, J.H. Dinitz (eds.) Handbook of Combinatorial
  Designs, Second Edition, pp. 622--628. Chapman and Hall/CRC, Boca Raton, FL
  (2006)

\bibitem{P1F-K50-starter}
Ihrig, E.C., Seah, E.S., Stinson, D.R.: A perfect one-factorization of
  {$K_{50}$}.
\newblock Journal of Combinatorial Mathematics and Combinatorial Computing
  \textbf{1}, 217--219 (1987)

\bibitem{P-Code}
Jin, C., Jiang, H., Feng, D., Tian, L.: \mbox{{P}-Code}: A new \mbox{{RAID}-6}
  code with optimal properties.
\newblock In: Proceedings of the 23rd International Conference on
  Supercomputing ({ICS}'09), pp. 360--369. Yorktown Heights, NY (2009)

\bibitem{P1F-2-k2p}
Kobayashi, M.: On perfect one-factorization of the complete graph {$K_{2p}$}.
\newblock Graphs and Combinatorics \textbf{5}(1), 351--353 (1989)

\bibitem{P1F-K36-even-starter}
Kobayashi, M., Awoki, H., Nakazaki, Y., Nakamura, G.: A perfect
  one-factorization of {$K_{36}$}.
\newblock Graphs and Combinatorics \textbf{5}, 243--244 (1989)

\bibitem{DACO}
Li, M., Shu, J.: {DACO}: A high-performance disk architecture designed
  specially for large-scale erasure-coded storage systems.
\newblock IEEE Transactions on Computers \textbf{59}(10), 1350--1362 (2010)

\bibitem{cyclic-code-starter}
Li, M., Shu, J.: On cyclic lowest density {MDS} array codes constructed using
  starters.
\newblock In: Proceedings of the 2010 {IEEE} International Symposium on
  Information Theory ({ISIT} '10), pp. 1315--1319. Austin, TX (2010)

\bibitem{equivalence}
Li, M., Shu, J.: On the equivalence between the \mbox{B-Code} constructions and
  perfect one-factorizations.
\newblock In: Proceedings of the 2010 IEEE International Symposium on
  Information Theory (ISIT 2010), pp. 993--996. Austin, TX (2010)

\bibitem{SDC}
Li, M., Shu, J.: Preventing silent data corruptions from propagating during
  data reconstruction.
\newblock IEEE Transactions on Computers \textbf{59}(12), 1611--1624 (2010)

\bibitem{GRID}
Li, M., Shu, J., Zheng, W.: {GRID} {Codes}: Strip-based erasure codes with high
  fault tolerance for storage systems.
\newblock ACM Transactions on Storage \textbf{4}(4), 1--22 (2009)

\bibitem{P1F-list-(even)-starter-K32}
Pike, D.A., Shalaby, N.: \mbox{Non-isomorphic} perfect
  \mbox{one-factorizations} from {Skolem} sequences and starters.
\newblock Journal of Combinatorial Mathematics and Combinatorial Computing
  \textbf{44}, 23--32 (2003)

\bibitem{disk_failure1}
Pinheiro, E., Weber, W.D., Barroso, L.A.: Failure trends in a large disk drive
  population.
\newblock In: Proceedings of the 5th {USENIX} conference on File and Storage
  Technologies ({FAST}'07), pp. 17--28. San Jose, CA (2007)

\bibitem{MinimumDensityRAID6Codes}
Plank, J.S., Buchsbaum, A.L., {Vander Zanden}, B.T.: Minimum density {RAID-6}
  codes.
\newblock ACM Transactions on Storage \textbf{6}(4), 1--22 (2011)

\bibitem{dRAID2-IBM}
Rao, K.K., Hafner, J.L., Golding, R.A.: Reliability for networked storage
  nodes.
\newblock IEEE Transactions on Dependable and Secure Computing \textbf{8}(3),
  404--418 (2011)

\bibitem{RS}
Reed, I.S., Solomon, G.: Polynomial codes over certain finite fields.
\newblock Journal of the Society for Industrial and Applied Mathematics
  \textbf{8}(2), 300--304 (1960)

\bibitem{RS-Cauchy2}
Roth, R.M., Lempel, A.: On {MDS} codes via {Cauchy} matrices.
\newblock IEEE Transactions on Information Theory \textbf{35}(6), 1314--1319
  (1989)

\bibitem{RS-Cauchy1}
Roth, R.M., Seroussi, G.: On generator matrices of {MDS} codes.
\newblock IEEE Transactions on Information Theory \textbf{31}(6), 826--830
  (1985)

\bibitem{disk_failure2}
Schroeder, B., Gibson, G.A.: Disk failures in the real world: What does an
  {MTTF} of 1,000,000 hours mean to you?
\newblock In: Proceedings of the 5th {USENIX} conference on File and Storage
  Technologies ({FAST}'07), pp. 1--16. San Jose, CA (2007)

\bibitem{P1F-K36-starter}
Seah, E.S., Stinson, D.R.: A perfect one-factorization for {$K_{36}$}.
\newblock Discrete Mathematics \textbf{70}(2), 199--202 (1988)

\bibitem{P1F-K40-even-starter}
Seah, E.S., Stinson, D.R.: A perfect one-factorization for {$K_{40}$}.
\newblock Congressus Numerantium \textbf{68}, 211--213 (1989)

\bibitem{RAID-high-level-survey}
Thomasian, A., Blaum, M.: Higher reliability redundant disk arrays:
  Organization, operation, and coding.
\newblock ACM Transactions on Storage \textbf{5}(3), 1--59 (2009)

\bibitem{P1F}
Wallis, W.D.: \mbox{One-Factorizations}.
\newblock Kluwer, Norwell, MA (1997)

\bibitem{dRAID1-IBM}
Wilcke, W.W., Garner, R.B., Fleiner, C., Freitas, R.F., Golding, R.A., Glider,
  J.S., Kenchammana-Hosekote, D.R., Hafner, J.L., Mohiuddin, K.M., Rao, K.,
  Becker-Szendy, R.A., Wong, T.M., Zaki, O.A., Hernandez, M., Fernandez, K.R.,
  Huels, H., Lenk, H., Smolin, K., Ries, M., Goettert, C., Picunko, T., Rubin,
  B.J., Kahn, H., Loo, T.: {IBM} {Intelligent} {Bricks} project---{Petabytes}
  and beyond.
\newblock IBM Journal of Research and Development \textbf{50}(2/3), 181--197
  (2006)

\bibitem{P1F-K52-even-starter}
Wolfe, A.J.: A perfect one-factorization of {$K_{52}$}.
\newblock Journal of Combinatorial Designs \textbf{17}(2), 190--196 (2009)

\bibitem{B-Code}
Xu, L., Bohossian, V., Bruck, J., Wagner, D.G.: Low-density {MDS} codes and
  factors of complete graphs.
\newblock IEEE Transactions on Information Theory \textbf{45}(6), 1817--1826
  (1999)

\bibitem{X-Code}
Xu, L., Bruck, J.: \mbox{{X}-Code}: {MDS} array codes with optimal encoding.
\newblock IEEE Transactions on Information Theory \textbf{45}(1), 272--276
  (1999)

\bibitem{ZZS}
Zaitsev, G.V., Zinov'ev, V.A., Semakov, N.V.: Minimum-check-density codes for
  correcting bytes of errors, erasures, or defects.
\newblock Problems of Information Transmission \textbf{19}(3), 197--204 (1983)

\end{thebibliography}

\end{document}